\documentclass[sigconf]{acmart}

\usepackage{amsmath,amsthm}
\usepackage{latexsym,amscd}
\usepackage{amsbsy}
\usepackage{lipsum}
\usepackage{algorithm}
\usepackage{algorithmic}
\usepackage{threeparttable}
\usepackage{color,cases}
\usepackage{bm}
\usepackage{relsize}
\usepackage{mathtools}
\usepackage{mathrsfs}  
\usepackage{url}
\usepackage{setspace}
\usepackage{graphicx}

\usepackage{hyperref}       
\hypersetup{
    colorlinks=true,       
    linkcolor=blue,    
     urlcolor=cyan
}

\newcommand{\red}[1]{\textcolor{red}{#1}}

\newtheorem{assumption}{Assumption}[section]

\newtheorem{criterion}[theorem]{Criterion}
\newtheorem{remark}[theorem]{Remark}

\newcommand{\eat}[1]{}

\newcommand{\bo}[1]{\textcolor{blue}{Bo: #1}}

\newcommand{\YH}[1]{\textcolor{red}{YH: #1}}

\newcommand\sign{\operatorname{sign}}
\DeclareMathOperator{\Tr}{Tr}
\DeclareMathOperator{\Var}{Var}

\renewcommand{\paragraph}[1]{\noindent\textbf{#1}}

\setcopyright{none} 
\acmConference[Anonymous Submission to ACM CCS 2022]{ACM Conference on Computer and Communications Security}{Due 15 May 2019}{London, TBD}
\acmYear{2022}

\begin{document}
\title{Is Vertical Logistic Regression Privacy-Preserving? \\ A Comprehensive Privacy Analysis and Beyond} 

\author{Yuzheng Hu}
\affiliation{UIUC}
\email{yh46@illinois.edu}

\author{Tianle Cai}
\affiliation{Princeton University}
\email{tianle.cai@princeton.edu}

\author{Jinyong Shan}
\affiliation{Sudo Technology Co.,LTD}
\email{shanjy@sudoprivacy.com}

\author{Shange Tang}
\affiliation{Princeton University}
\email{shangetang@princeton.edu}

\author{Chaochao Cai}
\affiliation{Sudo Technology Co.,LTD}
\email{caic@sudoprivacy.com}

\author{Ethan Song}
\affiliation{Sudo Technology Co.,LTD}
\email{songym@sudoprivacy.com}

\author{Bo Li}
\affiliation{UIUC}
\email{lbo@illinois.edu}

\author{Dawn Song}
\affiliation{UC Berkeley}
\email{dawnsong@cs.berkeley.edu}
\date{}

\begin{abstract}

    We consider vertical logistic regression (VLR) trained with mini-batch gradient descent --- a setting which has attracted growing interest among industries and proven to be useful in a wide range of applications including finance and medical research. We provide a comprehensive and rigorous privacy analysis of VLR in a class of open-source Federated Learning frameworks, where the protocols might differ between one another, yet a procedure of obtaining local gradients is implicitly shared. We first consider the honest-but-curious threat model, in which the detailed implementation of protocol is neglected and only the shared procedure is assumed, which we abstract as an oracle. We find that even under this general setting, single-dimension feature and label can still be recovered from the other party under suitable constraints of batch size, thus demonstrating the potential vulnerability of all frameworks following the same philosophy. Then we look into a popular instantiation of the protocol based on Homomorphic Encryption (HE). We propose an active attack that significantly weaken the constraints on batch size in the previous analysis via generating and compressing auxiliary ciphertext. To address the privacy leakage within the HE-based protocol, 
    we develop a simple-yet-effective countermeasure based on Differential Privacy (DP), and provide both  utility and privacy guarantees for the updated algorithm.  Finally, we empirically verify the effectiveness of our attack and defense on benchmark datasets. Altogether, our findings suggest that all vertical federated learning frameworks that solely depend on HE might contain {severe} privacy risks, and DP, which has already demonstrated its power in horizontal federated learning, can also play a crucial role in the vertical setting, especially when coupled with HE or secure multi-party computation (MPC) techniques. 
    
\end{abstract}

\ccsdesc[500]{Security and privacy~Privacy-preserving protocols}
\ccsdesc[500]{Security and privacy~Privacy protections}
\ccsdesc[500]{Computing methodologies~Machine learning}

\keywords{Vertical logistic regression, federated learning, homomorphic encryption, differential privacy, label recovery attack} 

\maketitle

\section{Introduction} \label{sec_intro}


Recent years have witnessed the advancement of Artificial Intelligence (AI) in a wide range of tasks, such as image recognition \cite{he2016deep}, natural language processing \cite{devlin2018bert}, recommendation \cite{gharibshah2020deep} and game \cite{silver2016mastering}.  These achievements crucially rely on the availability of massive amount of labeled data. Across various industries, however, data are usually distributed over multiple places in different forms, {and} most organizations only possess partial or limited amount of \textit{unlabeled} data due to the heavy cost of annotation. Further, it has been increasingly difficult for organizations to share and cooperatively use their data from a legislative perspective owing to regulations like General Data Protection Regulation (GDPR) \cite{voigt2017eu} and California Consumer Privacy Act (CCPA) \cite{CCPA}, which enforce strict requirements on the use of data in order to protect individuals' privacy. 

To address this issue, the concept of federated learning (FL) was proposed by Google \cite{konevcny2016federated, konevcny2016federated_kmy, mcmahan2016federated} in 2016, aiming to build machine learning models across multiple decentralized parties without revealing their raw data. Over the past few years, there is an increasing concern towards privacy and security within both academia and industry, and FL has been widely adopted in finance \cite{liu2020secure}, medical research \cite{jochems2016distributed}, smart phones \cite{hard2018federated} and so on. Meanwhile, several open source FL systems have been developed, including TensorFlow Federated (TFF) \cite{TFF} by Google, Federated AI Technology Enabler (FATE) \cite{FATE} by Tencent Webank, PySyft \cite{ryffel2018generic} by Openmined, LEAF \cite{caldas2018leaf} by CMU and PaddleFL \cite{PaddleFL} by Baidu. The convenience brought by these frameworks makes FL the mainstream paradigm for proprietary entities to learn jointly across their geographical boundaries.

Based on the distribution characteristics of data, FL can be categorized 
into horizontal federated learning (HFL) and vertical federated learning (VFL) \cite{yang2019federated}. HFL is a system in which all participants share the same feature space, but the sample space could be significantly different. For example, a set of domestic banks may share a similar financial profile, however very few common customers. In such case, HFL can be applied to agglomerate models learnt in each individual banks. In VFL (a.k.a. split learning), participants first align their user ID to ensure a common sample space, then jointly train a model based on the aggregation of different features. For example, banks and fund companies may own different perspectives of a customer’s financial status, so they can perform better risk management leveraging the information from multiple dimensions.

While FL came up as a promising solution to preserve users’ privacy, potential risk of data leakage has been recognized under this new paradigm, which may hinder its widespread adoption. The main observation is that, although the raw data stays untouched, the exchange of intermediate values within the execution of FL framework, might reveal much private information \textit{even with cryptography algorithm}. Zhu and Han~\cite{zhu2020deep} show that it is possible to \textit{learn from gradients} and recover the original input as long as the batch size is small; Zhao et al.~\cite{zhao2020idlg} further show that the label can be directly inferred from the gradients. Other forms of attack include \textit{membership inference} \cite{nasr2019comprehensive}, in which the adversary aims to learn whether a given sample belongs to another participant; \textit{property inference} \cite{melis2019exploiting}, where the adversary intends to learn data properties that are independent of the main task; and \textit{model poisoning/backdoor attack} \cite{bhagoji2019analyzing, bagdasaryan2020backdoor}, with the goal of causing the global model to misclassify a set of chosen inputs. As FL is gaining more popularity within the AI community, privacy analysis of the current FL frameworks has become one of the most important research directions, which could potentially lead to the improvement over existing algorithms and protocols.

It has come to our attention that a major study of the privacy analysis of FL focuses on HFL, whereas the corresponding work on VFL lacks, albeit VFL is being widely used in industry \cite{wang2019interpret, feng2020multi, gu2020federated, Webank1, Webank2}. As far as we know, there are only a few exceptions: Weng et al.~\cite{weng2020privacy} propose \textit{reverse multiplication attack} and \textit{reverse sum attack} for logistic regression and Secure XGBoost \cite{cheng2021secureboost} in the VFL setting and demonstrate the privacy leakage of feature and partial order; Luo et al.~\cite{luo2021feature} study feature attack with similar models in the prediction phase. However, the theoretical analyses are either not rigorous or based on unrealistic assumptions; further, none of these works consider the privacy of \textit{labels} of training instances, which is a crucial part of the raw data. Recently, Liu et al.~\cite{liu2020backdoor} point out that the label can be inferred from the sign of gradients, but this is just a simple extension of ~\cite{zhao2020idlg} in the vertical setting and does not hold when using an approximate version of the cross-entropy loss. Li et al.~\cite{li2022label} consider the model-agnostic split learning, and propose a norm-based approach and a direction-based approach to identify the labels based on the information of per-sample gradient during training; however, their analyses relies on certain heuristics regarding the behavior of the model and is restricted to the two-party case.

Aiming to provide an in-depth understanding of a \textit{specific} model, we focus on a popular setting in VFL, where two (or multiple) parties collaboratively train a logistic regression model using mini-batch gradient descent, with separated features and labels owned by only one of them \cite{hardy2017private}. We refer to the setting as vertical logistic regression (VLR). Though being simple, VLR has attracted growing interest among industries and proven to be useful in many practical settings such as risk control \cite{chen2021homomorphic}, credit scoring, \cite{zheng2020vertical} and clinical diagnosis \cite{wu2013privacy}.  Also, while this setting has been studied by \cite{weng2020privacy, luo2021feature} as mentioned above, our work tends to provide a  \textit{comprehensive} and \textit{rigorous} understanding of the training framework, which could potentially serve as a useful guideline for practice.

Specifically, starting from the two-party scenario and assuming an oracle of obtaining local gradients which is actually being deployed in practice \cite{FATE-HeterLR}, we provide a fine-grained privacy analysis for both parties, and propose attacks on both feature and label levels under suitable conditions of batch size. For feature attack we come up with a corresponding hardness result, and for label recovery attack we show it naturally generalizes to the multi-party scenario. We then take into account the protocol of obtaining gradients based on Homomorphic Encryption (HE), and propose an active attack via generating and compressing auxiliary ciphertext, which allows us to relax the constraints of batch size. See Figure \ref{fig:illustration} for an illustration of our privacy analysis.
We proceed to develop a countermeasure based on Differential Privacy (DP) against the proposed attack, and provide utility and privacy guarantees for the updated algorithm. Finally, we verify the effectiveness of our attack and defense on benchmark datasets, indicating the fundamental weakness in existing frameworks as well as the power of coupling DP with HE techniques in VLR.  

\begin{figure}[htbp]
\includegraphics[width=0.5\textwidth]{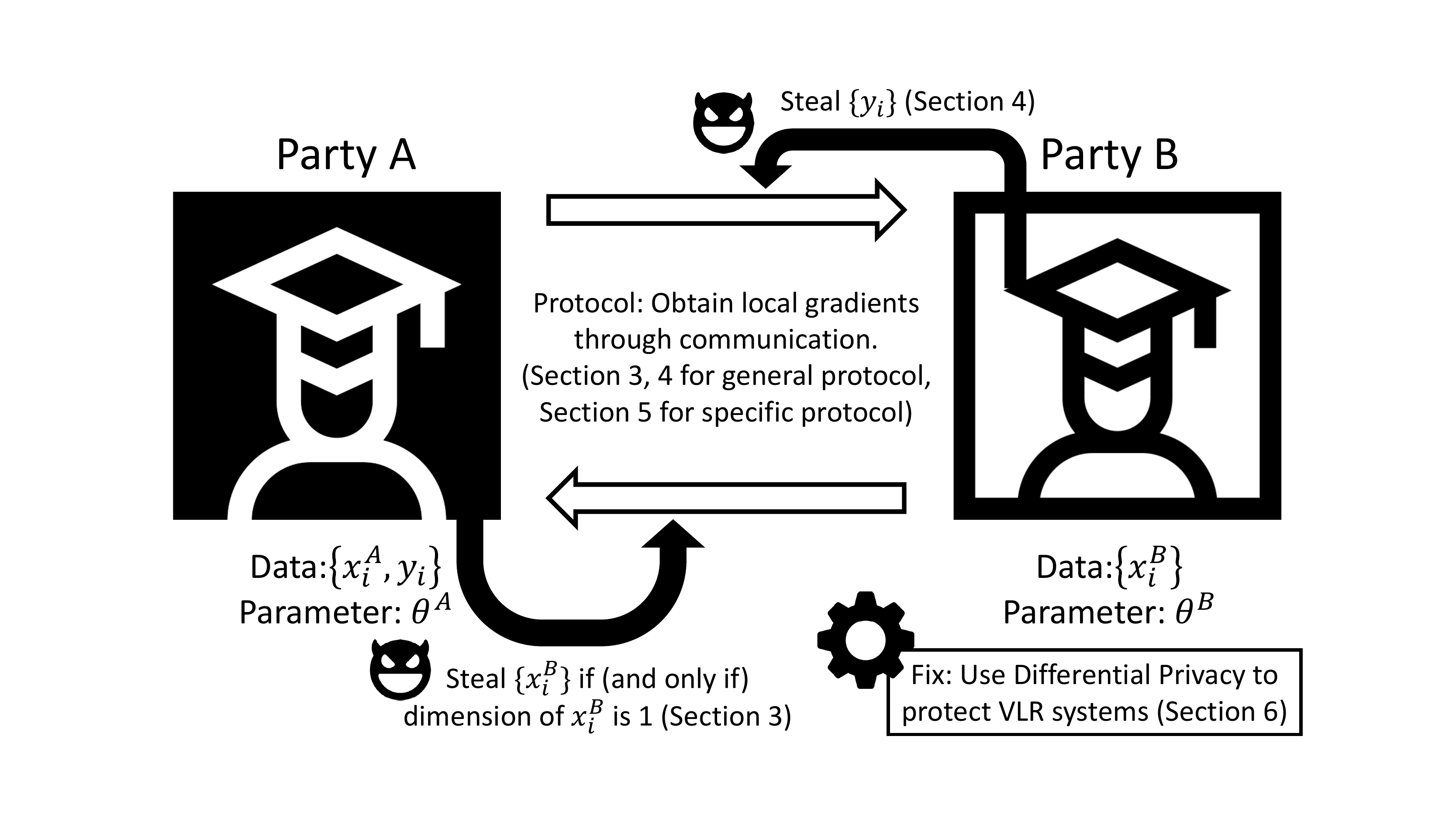}
\caption{Illustration of our privacy analysis. Under a general oracle of obtaining local gradients and some constraints on batch size, Party $A$ can steal $B$'s single-dimension feature, whereas $B$ can steal $A$'s labels. Within a specific HE-based protocol, both parties can perform the same attacks with less constraints on batch size. Differential Privacy serves as a simple-yet-effective countermeasure to the proposed attacks.}
\label{fig:illustration}
\end{figure}

The main contributions of this paper are summarized as follows:
\begin{itemize}
    \item We provide a fine-grained privacy analysis of VLR in the honest-but-curious setting: under suitable conditions of batch size, 1) for Party $A$, we construct a feature attack when $B$ has single-dimension feature and provide a hardness result for the rest case; 2) for Party $B$, we identify a label recovery attack when the initialization is small. We also show that the label recovery attack can naturally generalize when there are multiple parties. 
    {Our attack is strong in the sense that we (i) assume no \textit{prior knowledge of the data}; (ii) make no assumptions on the \textit{model's accuracy}; (iii) do no resort to external tools such as \textit{learning-based algorithm}; it holds (iv) independently of the \textit{existence of the third party} and (v) can be performed at \textit{any time} during training}. It only requires the knowledge of intermediate gradients, and is formalized via a careful dissection of the linear system;
    \item We then consider the active adversarial model: by taking into account the specific implementation of protocol based on Homomorphic Encryption (HE), we propose an active attack via generating and compressing auxiliary ciphertext, which allows us to relax the constraints of batch size and greatly extend the range of applications;
    \item We develop a simple-yet-effective countermeasure for the proposed attack based on Differential Privacy (DP), which is realized by injecting Gaussian noise to the HE-based protocol. While DP appears as a standard building block in HFL \cite{geyer2017differentially, truex2019hybrid, wei2020federated}, to the best of our knowledge, we are the first to couple it with HE techniques in the vertical setting. We also provide both utility and privacy guarantees for the updated algorithm;
    \item Experiment results on benchmark datasets verify the effectiveness of our proposed attack and defense, and showcase both the potential weakness of general VFL frameworks that solely rely on HE, as well as the capability of coupling DP with HE or MPC techniques to achieve better performance in terms of data privacy.  
\end{itemize}

The rest of the paper is organized as follows: Section \ref{sec_prelim} is our preliminary. We provide privacy analysis for Party $B$ and $A$ in Section \ref{sec_security_A} and \ref{sec_security_B} respectively. We then introduce the protocol of obtaining gradients and propose an active attack in Section \ref{sec_discuss}. Section \ref{sec_defense} describes our countermeasure and provides theoretical guarantees for both privacy and utility. In Section \ref{sec_experiment} we conduct experiments on benchmark datasets to evaluate our attack and defense. We conclude our work in Section \ref{sec_conclusion}.

\subsection{Related Works}

\paragraph{Privacy Leakage of Federated Learning.}
Identifying the privacy leakage of FL has seen a growing interest over the past years. Among them, a major set of works consider the setting of HFL, while only a few exceptions involve the privacy analysis of VFL \cite{weng2020privacy, luo2021feature, liu2020backdoor, li2022label}. Within HFL, \textit{data leakage} attracts the most attention: \cite{hitaj2017deep, wang2019beyond} use GAN \cite{NIPS2014_5ca3e9b1} to generate class-representative and synthesized samples; \cite{he2019model, zhu2020deep, geiping2020inverting, yin2021see, erdogan2021unsplit} adopt a optimization framework to invert the gradient and \textit{learn} the data; \cite{aono2017privacy} observes the special structure of gradient with single input and partially recovered the private data; \cite{bhagoji2019analyzing, li2019quantification} show how to reconstruct the data when they are uniformly distributed or binary; \cite{zhao2020idlg} further demonstrates that the sign of cross-entropy loss could directly reveal the label. Meanwhile, other forms of attack have been proposed, including \textit{membership inference} \cite{shokri2017membership, nasr2019comprehensive}, \textit{property inference} \cite{melis2019exploiting} which resembles the class-representative inference in \textit{model inversion} \cite{fredrikson2015model}, and \textit{model poisoning/backdoor attack} \cite{baruch2019little, bagdasaryan2020backdoor}. We refer the reader to \cite{bouacida2021vulnerabilities} for a comprehensive survey of the security vulnerabilities in  FL ecosystem, and \cite{kairouz2019advances} for general open problems in FL.

\medskip

\paragraph{General Defense Strategies.}
\begin{itemize}
    \item \textbf{Homomorphic Encryption.} Homomorphic Encryption (HE) \cite{rivest1978method} is a form of encryption which allows operations on the encrypted data without access to the secret key. Based on the utility of the cryptosystem, HE can be categorized into partially homomorphic encryption (PHE), somewhat homomorphic
    encryption (SWHE) and fully homomorphic encryption (FHE) \cite{acar2018survey}. While FHE enjoys more desirable properties and has been deployed in practical settings such as training with genomic data \cite{chen2018logistic}, we resort to the Paillier cryptosystem for efficiency, which is a standard form of additively homomorphic encryption within PHE and has been used in VLR in many previous works \cite{aono2016scalable, aono2017privacy, hardy2017private}. 

    \item \textbf{Differential Privacy.} Differential Privacy (DP) \cite{dwork2006differential} offers a probabilistic framework to \textit{quantify} privacy and has become a standard building block in machine learning community ever since its birth. DP ensures that the outcome of an algorithm will not be impacted by the presence or absence of a single instance, and is usually realized via \textit{randomization}: making the algorithm stochastic by adding noise.  Based on the structure of noise, there are several types of mechanism to achieve DP: Laplace, Exponential and Gaussian Mechanism \cite{dwork2014algorithmic}. See Section \ref{sec_defense} for formal definitions. 
    
    DP has been applied to a wide range of learning tasks to ensure privacy \cite{chaudhuri2011differentially, song2013stochastic, bassily2014private, abadi2016deep}, in which the noise is injected to the parameter after a step of update. For FL, DP is mostly used in the horizontal setting \cite{geyer2017differentially, truex2019hybrid, wei2020federated}. On the contrary, we apply DP to the VFL framework, carefully couple the noise-injection mechanism with the HE-based protocol and inject noise to the sensitive information before communication.
\end{itemize}

\section{Preliminaries}  \label{sec_prelim}
\subsection{(Vertical) Logistic Regression}\label{subsec:LR}

In binary logistic regression, the goal is to learn a linear model $\bm{\theta}\in \mathbb{R}^d$ that correctly maps the sample $\bm{x}\in \mathbb{R}^d$ to its corresponding label $y\in \{-1,1\}$. To achieve this, we use the training set $D=\{(\bm x_i,y_i)\}_{{i\in [n]}}$ {with $[n]=\{1,2,\ldots,n\}$}, and consider the following unregularized empirical risk minimization:
\begin{align}
    \min_{\bm{\theta}} \ell_{D}(\bm{\theta}) :=  -\frac{1}{n}\sum_{i=1}^n\log S(y_i\bm{\theta}^{\top}\bm{x}_i), \label{eq.emp}
\end{align}
where $S(z)=\frac{1}{1+e^{-z}}$ is the sigmoid function. In practice, we usually resort to mini-batch gradient descent, which is an effective algorithm of solving Eq. (\ref{eq.emp}).

\paragraph{Mini-batch Gradient Descent (MGD).}
The parameter is initialized as $\bm\theta = \bm\theta_0$. At $t$-th iteration, a mini-batch $\mathcal{B}_t$ of the indices of training examples with fixed batch size $s$ is sampled from $D$, and the mini-batch gradient, which is the average gradient over $\mathcal{B}_t$, is computed as follows: 
\begin{align}
\nabla\ell_{\mathcal{B}_t}(\bm\theta_{t-1})=\frac{1}{s}\sum_{i\in \mathcal{B}_t} \left(S(y_i\bm{\theta}_{t-1}^{\top}\bm{x}_i)-1\right)y_i\bm{x}_i.\label{MGD}
\end{align}
The model is then updated as 
\begin{align}\label{SGD}
    \bm\theta_t = \bm\theta_{t-1} - \eta_t \nabla\ell_{\mathcal{B}_t}(\bm\theta_{t-1}),
\end{align}
where $t\geq 1$ and $\eta_t$ represents the {learning rate}. We assume $s$ divides $n$, so there are $m := n/s$ iterations in one epoch. The algorithm is executed for $e$ epochs and there will be $T := me$ iterations in total.

\paragraph{First-order Taylor Approximation.}
Since the sigmoid function is non-polynomial, and the operations of exponential and logarithmic functions are cryptographically computational-inefficient, a natural solution is to use simple functions to approximate the gradient. There are several standard methods to approximate Eq. (\ref{MGD}), such as first-order Taylor approximation \cite{aono2016scalable, hardy2017private}, minimax approximation \cite{chen2018logistic} and piecewise approximation \cite{mohassel2017secureml}. {It turns out that the specific choice of approximation scheme has minimal effect on the privacy analysis. Therefore, we will focus on the first one throughout the main body of this paper, and leave the analyses of other approximation methods to Appendix \ref{app_minimax} and \ref{app_piecewise}.}

 Now, doing Taylor's expansion of $S(z)$ at $z=0$, we have
\begin{align}
   S(z) = \frac{1}{2} + \frac{1}{4}z + O(z^2).    \nonumber 
\end{align}
Therefore, the gradient in Eq. (\ref{MGD}) can be approximated as \begin{align}\label{eq MGD-taylor}
    \nabla\ell_{\mathcal{B}_t}(\bm\theta_{t-1})\approx \frac{1}{4s}\sum_{i\in \mathcal{B}_t} f_{i,t}\bm{x}_i,
\end{align} 
where we define the coefficients 
\begin{align}
f_{i,t} := \bm\theta^{\top}_{t-1}\bm x_i-2y_i, \ \forall t {\in [T]}, \ \forall i\in[n].
\end{align}

\paragraph{Vertical Logistic Regression.} We start with two-party vertical logistic regression (the multi-party case will be discussed at the end of Section \ref{sec_security_B}). Suppose $A$ and $B$ vertically share a group of data --- formally, given a sample $(\bm{x},y)\in \mathbb{R}^d\times\{-1,1\}$, the feature $\bm x$ is split into $\bm{x}^A \in \mathbb{R}^{d_A}$ and $\bm{x}^B \in \mathbb{R}^{d_B}$, where $\bm{x} = \left(\begin{array}{@{}c@{}}
     \bm{x}^A  \\
     \hline
     \bm{x}^B
\end{array}\right)
$, $d = d_A + d_B$, $(\bm x^A,y)$ belongs to Party $A$, and $\bm x^B$ belongs to Party $B$. Now, both parties  want to jointly optimize a binary logistic regression model using mini-batch gradient descent. Since we are considering a linear model, the separation in feature naturally induces a separation in parameter, i.e. the model parameter $\bm \theta\in \mathbb{R}^d$ is also split into $\bm \theta^A\in \mathbb{R}^{d_A}$ and $\bm \theta^B\in \mathbb{R}^{d_B}$, belonging to Party $A$ and Party $B$ respectively. The gradients of both parties at each iteration then write as
\begin{align}\label{eq:own_gradient}
\nabla^A_{t}=\frac{1}{4s}\sum_{i \in \mathcal{B}_t}f_{i, t}\bm x^A_{{i}}\quad \text{and} \quad \nabla^B_{t}=\frac{1}{4s}\sum_{i\in \mathcal{B}_t}f_{i, t}\bm x^B_{{i}}.
\end{align} 

\paragraph{Data Assumption.} We assume that the data $\{\bm x_i\}_{i\in [n]}$ are sampled from a probability distribution which is continuous w.r.t. the Lebesgue measure over $\mathbb{R}^{d\times n}$, and that $\|\bm x_i\|_2\le 1, \ \forall i \in [n]$. These assumptions are canonical in the machine learning literature.

\subsection{{Threat Models}} \label{subsection pf}

In this subsection, we formally state the key ingredients of our threat models. We use two parties $A$ and $B$ for illustration and consider both the \textit{passive} adversarial model (see Section \ref{sec_security_A} and \ref{sec_security_B}) and the \textit{active} adversarial model (see Section \ref{sec_discuss}).

\paragraph{Adversary’s objective.} 
At a time during training, the main goal of the label party (Party $A$) is to identify the feature of the non-label party (Party $B$); whereas the main goal of a non-label Party is to recover the hidden label from the label party.

\paragraph{Passive Adversarial Model.} In the passive adversarial model (a.k.a. the honest-but-curious model), a corrupted party (Party $A$ or Party $B$) correctly follows the protocol specification and the adversary obtains the input and internal state of the corrupted party~\cite{hazay2010efficient}. We make two key assumptions within this threat model: 1) first of all, each party is capable of obtaining his own share of gradient ($\nabla^A_{t}$ or $\nabla^B_{t}$) at each iteration --- we abstract such ability as an oracle and ignore the details of implementation for now; 2) second, the non-label party holds the information about the initialization scheme (not necessarily the precise value), i.e. how $\bm \theta_0$ is initialized. 

\paragraph{Active Adversarial Model.} In the active adversarial model, a corrupted party can arbitrarily deviate from the protocol specification, according to the adversary’s instructions~\cite{hazay2010efficient}. In this adversarial model, we focus on an \textit{instantiation} of computing the gradients $\nabla^A_{t}$ and $\nabla^B_{t}$ based on HE \cite{FATE-HeterLR,hardy2017private}, and show how the adversary can obtain more information by modifying the original protocol. Such active attack is expected to be stronger than the ones proposed in the passive adversarial model, as will be seen in Section \ref{sec_discuss}.

\paragraph{Publicly-available information.} The \textit{mini-batch} and the \textit{learning rate} used in each iteration are known to all participants --- combined with the assumption in the passive adversarial model, this ensures that both parties can do local updates on their own share of model parameters.

\section{Privacy Analysis for Party $B$} \label{sec_security_A}
As an adversary, Party $A$ seeks to obtain information about the feature $\bm x_{B}$. At $t$-th iteration, $A$ receives his own gradient
\begin{align} \label{gradient A}
\nabla^A_{t}=\frac{1}{4s}\sum_{i\in \mathcal{B}_t}f_{i, t}\bm x^A_{i},
\end{align}
which is the linear combination on his own share of feature. We will restrict our discussion to $s\le d_A$, and resort to a useful result in linear algebra.

\begin{theorem}  \label{Theorem Linear}
Suppose $\{\bm z_i\}_{i\in [n]}$ is randomly sampled from a probability distribution which is continuous w.r.t. the Lebesgue measure over $\mathbb{R}^{d'\times n}$ and we want to solve $a_i \in \mathbb{R}$, then $\forall I \in \mathbb{R}^{d'}, $  $\forall S \subseteq D$ such that $|S| = s$, the linear system 
\begin{align}
    \sum_{i\in S} a_i\bm z_i = \bm I 
\end{align}
has at most one solution with probability one if $s \le d'$. 
\end{theorem}
{The proof is mainly based on the fact that degenerate matrix has zero Lebesgue measure, and the details are deferred to Appendix \ref{app_sub_A}.} As a direct consequence, we immediately have the following corollary. 

\begin{corollary} \label{corollary unique}
If $s \le d_A$, then $\forall t$, Party $A$ can uniquely determine the values of $f_{i, t}$ for all $i \in {\mathcal{B}}_t$ with probability one.
\end{corollary}

\begin{proof}
The uniqueness follows from Theorem \ref{Theorem Linear}, and the fact that $\{\bm x_{A,i}\}_{i\in[n]}$ also follows a continuous probability distribution. The existence is trivial.
\end{proof}

{Saying in another way, Corollary \ref{corollary unique} allows Party $A$ to determine the coefficients of the linear system Eq. (\ref{gradient A})}. Now denote 
\begin{align}
    g_{i,t}:= \bm\theta_{t-1}^{\top}\bm x_{i}
\end{align}
and decompose
\begin{align}
    g_{i,t} = \underbrace{(\bm\theta_{t-1}^A)^{\top}\bm x^A_{i}}_{g^A_{i,t}} + \underbrace{(\bm\theta_{t-1}^B)^{\top}\bm x^B_{i}}_{g^B_{i,t}}.
\end{align}
After obtaining the coefficient
\begin{align*}
    f_{i, t}=g_{i,t}-2y_i,
\end{align*}
Party $A$ can extract the information known to himself (which is $ g^A_{i,t}-2y_i$), and obtain the values of $g_{i,t}^B$, which is exactly the inner product of the data $\bm x$ and the parameter $\bm\theta$ on $B$'s feature. Therefore, in the special case $d_B = 1$, Party $A$ can determine the \textit{order} of the data according to the $x_B$ feature in each batch. 

\begin{lemma}\label{lem order}
Suppose $s \le d_A$ and $d_B = 1$, then Party $A$ can learn the ratio of the data in ${\mathcal{B}}_t$ in terms of feature $x^B$. 
\end{lemma}

\begin{proof}
For $i, j\in {\mathcal{B}}_t$, Party $A$ can compute the ratio of $x_{B,i}$ and $x_{B,j}$ via calculating $g_{i,t}^B/g_{j,t}^B$. Therefore, he will be able to sort the data in ${\mathcal{B}}_t$ in terms of feature $x^B$.
\end{proof}

There are two special cases of Lemma \ref{lem order}: when $s=n$, Party $A$ directly obtains the order of the full dataset $D$; when $s=1$, there is essentially \textit{no} sorting and thus Party $A$ can obtain \textit{no} information about the feature $x_B$. When $2 \le s \le n-1$, in order to obtain the full order, we need a few epochs to make different batches `intertwine' with each other, as shown in the following lemma.

\begin{lemma} \label{lem inter}
    Under the same condition as in Lemma \ref{lem order}, and we further assume $s \ge 2$. Then Party $A$ can learn the ratio of the dataset $D$ in terms of feature $x_B$ in at most $\lfloor\log_2(m)\rfloor + 1$ epochs with probability at least $1-\mathcal{O}(\frac{1}{m^{s-1}})$. 
\end{lemma}

The proof of Lemma \ref{lem inter} is deferred to Appendix \ref{app_sub_A}. Now, after $\lfloor\log_2(m)\rfloor + 1$ epochs, Party $A$ will be able to learn the ratio between different data points on $B$'s feature: 
\begin{align}\label{eq ri}
    \frac{x^B_{i}}{x^B_{1}} := r_i, \ 2 \le i\le n.
\end{align}
Without loss of generality, assume $1 \in \mathcal{B}_1$ and $1 \in \mathcal{B}_k$ for some $m+1 \le k \le 2m$. By comparing $g_{1,k}$ and $g_{1,1}$ --- {the inner products of the same feature $x_1^B$ and $B$'s parameter $\theta^B$ at different iterations}, we immediately have the following theorem. 

\begin{theorem}\label{theorem precise}
Suppose $s \le d_A$ and $d_B = 1$, then Party $A$ can learn the precise value of the data in terms of feature $x^B$ (up to a difference in sign) in at most $\max\{\lfloor\log_2(m)\rfloor + 1,2\}$ epochs with probability at least $1-\mathcal{O}(\frac{1}{m^{s-1}})$. 
\end{theorem}

\begin{proof}
By the update formula of $\theta^B_{k}$, we have
\begin{align}
    g_{1,k} - g_{1,1} &= \theta^B_{k-1}x^B_{1} - \theta^B_{0}x^B_{1} \nonumber \\
    &=-\sum_{t=1}^{k-1} \frac{\eta_t}{4s}\sum_{j \in \mathcal{B}_t}(\theta_{t-1}{x}_j - 2y_j)x^B_{j}x^B_{1} \nonumber \\
    &=-\sum_{t=1}^{k-1} \frac{\eta_t}{4s}\sum_{j \in \mathcal{B}_t} f_{j,t}r_j(x^B_{1})^2 \nonumber \\
    &=-\left[\sum_{t=1}^{k-1}\left( \frac{\eta_t}{4s}\sum_{i \in D_t} f_{j,t}r_j\right)\right](x^B_{1})^2.
\end{align}
Therefore,
\begin{align}
     x^B_{1} = \pm\sqrt{\frac{4s(g_{1,1}-g_{1,k})}{\sum\limits_{t=1}^{k-1}\sum\limits_{j \in \mathcal{B}_t} f_{j,t}r_j\eta_t}}
\end{align}
and $x^B_i = r_ix^B_{1}$, where $r_i$ is defined in Eq. (\ref{eq ri}).
\end{proof}

When $d_B \ge 2$, we provide a hardness result: by simply observing $\{g_{i,t}^B\}_{i \in \mathcal{B}_t, t \in [T]}$, Party $A$ cannot learn $B$'s feature in general. {The main observation is that an orthogonal transformation would not change the coefficients obtained during training.} 
\begin{theorem} \label{theorem hardness}
Assume Party $A$ does not know the precise value of $\bm \theta^B_{0}$ and $d_B \ge 2$. Fix the mini-batch and learning rate used in each iteration, then there are infinitely many choices of $\{\bm x_i^B\}_{i \in [n]}$ that yield the same set of $\{g_{i,t}^B\}_{i \in \mathcal{B}_t, t \in [T]}$ through the optimization process. As a consequence, $A$ cannot determine the precise value of $B$'s feature.
\end{theorem}

\begin{proof}
Denote $O(d_B, \mathbb{R})$ as the orthogonal group in dimension $d_B$. Since $A$ does not hold the information of $\bm \theta_0^B$, we can pick an arbitrary $\bm Q \in O(d_B, \mathbb{R})$ and consider the following transformation
\begin{align*}
     \bm \theta_{t-1}^B \mapsto  \bm Q \bm \theta_{t-1}^B, \  \forall t\in [T] \ \  \text{and} \ \ \bm x^B_{i} \mapsto \bm Q\bm x^B_{i}, \  \forall i\in[n].
\end{align*}
By the update formula on $\bm\theta^B$, we have
\begin{align}
    \bm Q\bm\theta_t^B = \bm Q\bm\theta_{t-1}^B - \frac{1}{4s}\sum_{i \in \mathcal{B}_t}(\bm \theta_{t-1}^{\top} \bm x_i - 2y_i) \bm Q\bm x_i^B,
\end{align}
which is equivalent to
\begin{align}
    \bm Q\bm\theta_t^B = \bm Q\bm\theta_{t-1}^B - \frac{1}{4s}&\sum_{i \in \mathcal{B}_t}\big((\bm Q\bm\theta_{t-1}^B)^{\top} (\bm Q\bm x_i^B) \nonumber \\
    &+ (\bm \theta_{t-1}^A)^{\top}\bm x_i^A - 2y_i\big) \bm Q\bm x_i^B,
\end{align}
i.e. the mini-batch gradient descent acting on a new dataset with $
\hat{\bm {x}}_i = \left(\begin{array}{@{}c@{}}
     \bm x_i^A  \\
     \hline
     \bm Q \bm x_i^B
\end{array}\right)$ and starting from $\bm Q\bm \theta_0$.
Party $A$ cannot distinguish the two processes by observing $g_{i,t}^B$, since the inner product
\begin{align}
    \left<\bm Q \bm \theta_{t-1}^B, \bm Q\bm x^B_{i}\right> = \left<\bm \theta_{t-1}^B, \bm x_i^B\right>
\end{align}
remains unchanged. Combining the fact that $O(d_B, \mathbb{R})$ has infinite cardinality when $d_B \ge 2$ gives the result as desired.
\end{proof}

\section{Privacy Analysis for Party $A$}  \label{sec_security_B}
We shall see that Party $A$ and $B$ are symmetric except that $B$ does not hold the label. Therefore, the main goal of an adversarial $B$ is to obtain information about the label $y$. At $t$-th iteration, $B$ receives his own gradient
\begin{align} \label{gradient B}
   \nabla^B_{t}=\frac{1}{4s}\sum_{i\in \mathcal{B}_t}f_{i, t}\bm x^B_{i}
\end{align}
which is still the linear combination on his own share of feature. Similar to the previous section, we will restrict our discussion to $s\le d_B$, implying that the linear system given by Eq. (\ref{gradient B}) is not underdetermined. As a corollary of Theorem \ref{Theorem Linear}, Party $B$ also holds the value of 
\begin{align}
    f_{i,t} = (\bm\theta_{t-1}^A)^{\top}\bm x^A_{i}+(\bm\theta_{t-1}^B)^{\top}\bm x^B_{i}-2y_i
\end{align}
for $t \in [T]$ and $i\in \mathcal{B}_t$. In the right hand side, Party $B$ does not know about the first and the third term, so at first sight it seems unlikely for him to gain information about the label. However, we will show in the following analysis that label attack would still be possible under certain circumstances.

{We first consider a special case where Party $A$ \textit{only} holds the label, meaning that Party $B$ have full control of the model.
\begin{theorem}   \label{theorem extreme}
Suppose $d_A = 0$, then Party $B$ can uniquely determine the label $\{y_i\}_{i\in S_t}$ for arbitrary sampled batch $\mathcal{B}_t$ with probability one.  
\end{theorem}
}

{
The proof of Theorem \ref{theorem extreme} is deferred to Appendix \ref{app_sub_B}. Some might argue for a simpler attack in this special case: since $B$ holds both the data and the model, he can directly feed the data to the trained model and use the output to `guess' the label. However, there are two drawbacks of this approach: 1) it heavily relies on the model's training accuracy, and do not have any \textit{guarantees} on the predictions; 2) the attack can only be performed after (or at the end of) training.}       

Now we will consider the more realistic case $d_A \ge 1$. The intuition here is that $|2y_i| = 2$, so if Party $B$ can convince himself that the full inner product satisfies $|\bm\theta_{t-1}^{\top}\bm x_i| < 2$, then 
\begin{align}
    \sign(f_{i,t}) = \sign(-2y_i) = -y_i,
\end{align}
so he can use the opposite sign of $f_{i,t}$ to determine the value of $y_i$ in the $t$-th iteration. This can be formulated into the following criterion.

\begin{criterion}  \label{criterion}
For $t=1,\cdots, T^*$, Party $B$ uses $-\sign(f_{i,t})$ to determine $y_i, \  \forall i\in \mathcal{B}_t$. 
\end{criterion}

\begin{remark}
If Criterion \ref{criterion} is valid for $T^* \ge m$, then Party $B$ can obtain the label of the full dataset $D$. 
\end{remark}

\noindent Now the analysis naturally breaks into two parts:
\begin{itemize}
    \item Show when Criterion \ref{criterion} is valid for $T^*=1$;
    \item Determine the largest $T^*$ that we can \textit{safely} apply the criterion.
\end{itemize}
We will first deal with the second part. Define $h_t = \max\limits_{i} \{|\bm \theta_{t-1}^{\top}\bm x_i| + 2 \}$ and we will establish a recursive relation of $h_t$ in the following lemma. 

\begin{lemma}  \label{lemma}
Suppose we use a constant learning rate $\eta_t = \eta$ in mini-batch gradient descent. Then we have
\begin{align}
    h_{t+1} \le \left(1+\frac{\eta}{4}\right) h_{t}.
    \end{align}
\end{lemma}

\begin{proof}
By the update formula Eq. (\ref{SGD}), for $1\leq j \leq n$ we have
\begin{align}
    \bm\theta_t^{\top}\bm x_j = \bm \theta_{t-1}^{\top}\bm x_j - \frac{\eta}{4s} \sum_{i \in \mathcal{B}_t} (\bm \theta_{t-1}^{\top}\bm x_i - 2y_i)\bm x_i^{\top}\bm x_j.
\end{align}
Applying Cauchy-Schwarz and triangle inequality, we have
\begin{align}
    \big|\bm \theta_t^{\top}\bm x_j\big| \le \big|\bm \theta_{t-1}^{\top}\bm x_j\big| + \frac{\eta}{4}\cdot \frac{1}{s}\sum_{i \in \mathcal{B}_t}\left(\big|\bm \theta_{t-1}^{\top}\bm x_i\big|+2\right),
\end{align}
implying
\begin{align}
    \left(\big|\bm\theta_t^{\top}\bm x_j\big|+2\right) \le \left(\big|\bm\theta_{t-1}^{\top}\bm x_j\big|+2\right) + \frac{\eta}{4s}\sum_{i \in \mathcal{B}_t}\left(\big|\bm\theta_{t-1}^{\top}\bm x_i\big|+2\right).
\end{align}
Taking maximum over $1\leq j \leq n$ on both side, we immediately have
\begin{align}
    h_{t+1} \le \left(1+\frac{\eta}{4}\right) h_{t}.
\end{align}
\end{proof}

As a corollary, we immediately obtain a possible choice of $T$ that will be safe for the implementation of the criterion. 
\begin{corollary}  \label{corollary T}
Suppose $\max\limits_{i} \{|\bm\theta_0^{\top}\bm x_i|\} = \epsilon < 2$. Then Criterion \ref{criterion} is valid for
\begin{align}
    T_\epsilon = \bigg\lceil \log_{\left(1+\frac{\eta}{4}\right)}\left(\frac{4}{2+\epsilon}\right) \bigg\rceil.
\end{align}
\end{corollary}

\begin{proof}
We have $h_1 = 2+\epsilon$ and the criterion fails with threshold $4$. By Lemma \ref{lemma}, the criterion is safe for the first $\bigg\lceil \log_{\left(1+\frac{\eta}{4}\right)}\left(\frac{4}{2+\epsilon}\right) \bigg\rceil$ iterations.
\end{proof}

We now come back to the first part. Party $B$ need to determine whether Criterion \ref{criterion} is valid for $t=1$, \textit{i.e.} $\max\limits_{i} \{|\bm\theta_0^{\top}\bm x_i|\} < 2$ based on his knowledge of the initialization scheme. We will discuss several possible choices of initialization.

\medskip

    \paragraph{Deterministic Initialization.}
    Suppose $\bm\theta_0$ is initialized to be a fixed vector. If  
    \begin{align}
        \max\limits_{i} \{|\bm\theta_0^{\top}\bm x_i|\} = \epsilon < 2, 
    \end{align}
    then by Corollary \ref{corollary T}, we can set $T^*=T_\epsilon$ in Criterion \ref{criterion}. As a special case, we have
    \begin{theorem}
    If we apply \textit{Zero Initialization} \emph{\cite{tripathi2018zero}} to $\bm\theta_0$, then Criterion \ref{criterion} is valid for 
    \begin{align}
         T_0 = \bigg\lceil \log_{\left(1+\frac{\eta}{4}\right)} 2  \bigg\rceil.
    \end{align}

    \end{theorem}

  \paragraph{Small Random Initialization.} This is a common strategy in modern machine learning problems to ensure an $O(1)$ output, so that optimization can be significantly boosted. Common initialization schemes such as $\textit{Xavier Initialization}$ \cite{glorot2010understanding} and $\textit{Kaiming Initialization}$ \cite{he2015delving} all fall into this category. Formally, since $\|\bm x_i\|_2 \le 1, \ \forall i\in[n]$, we need to ensure that $\mathbb{E}\|\bm \theta_0\|_2 = O(1)$, $i.e.$ \ $\mathbb{E}\theta_{0,j}^2 = O\left(\frac{1}{d}\right), \ \ \forall j \in [d]$. To achieve this, we set $\gamma$ to be a mean-zero sub-exponential random variable
  and initialized each element of $\bm\theta_0$ \textit{i.i.d.} through $\gamma/\sqrt{d}$. We then have the following lemma.
        \begin{lemma} \label{lemma init}
            Suppose $\|\gamma\|_{\psi_1} = l$, then with probability at least $1-2ne^{-\sqrt{d}}$, we have
            \begin{align}
                \big|\bm\theta_0^{\top}\bm x_i\big| < \frac{l}{c}d^{-\frac{1}{4}}, \ \ \ \ \ \forall i\in [n],
            \end{align}
            where $c \le 1$ is an absolute constant.
        \end{lemma}
        
        \begin{proof}
            For a fixed $i \in [n]$, note $\|\bm x_i\|_2 \le 1$ and $\|\bm x_i\|_{\infty} \le 1$. Applying Bernstein's inequality (Theorem 2.8.1 in \cite{vershynin2018high}), we have
            \begin{align}
& \ \ \  \Pr\left(\bigg|\sum_{j=1}^d \theta_{0,j}x_{i, j}\bigg|\ge \epsilon\right) \nonumber \\
&\le 2\exp\left[-c\min\left(\frac{\epsilon^2}{L^2\|\bm x_i\|_2^2},\frac{\epsilon}{L\|\bm x_i\|_{\infty}}\right)\right] \nonumber \\
&\le 2\exp\left[-c\min\left(\frac{d\epsilon^2}{l^2},\frac{\sqrt{d}\epsilon}{l}\right)\right].  
\end{align}
Setting $\epsilon = \frac{l}{c}d^{-\frac{1}{4}}$, we have with probability at most $2e^{-\sqrt{d}}$, 
\begin{align}
    |\bm\theta_0^{\top}\bm x_i| \ge \frac{l}{c}d^{-\frac{1}{4}}.
\end{align}
Taking a union bound over $i \in [n]$ yields the result as desired. 
        \end{proof}

        Together with Corollary $\ref{corollary T}$, we immediately have
        \begin{theorem}  \label{theorem small}
        Suppose each element of $\bm\theta_0$ is initialized \textit{i.i.d.} through $\gamma/\sqrt{d}$, where $\gamma$ is a mean-zero sub-exponential random variable with $\|\gamma\|_{\psi_1} = l$. Then with probability at least $1-2ne^{-\sqrt{d}}$, Criterion \ref{criterion} is valid for 
        \begin{align}
            T_{l} = \bigg\lceil \log_{\left(1+\frac{\eta}{4}\right)}\left(\frac{4c}{2c+l d^{-\frac{1}{4}}}\right) \bigg\rceil,
        \end{align}
         where $c \le 1$ is an absolute constant.
        \end{theorem}
        
    \begin{remark}
    Suppose Party $B$ holds the label of the full dataset $D$ in the first epoch, then he can repeat the procedure as in Section \ref{sec_security_A} and obtains $g_{i,t}^A = (\bm\theta_{t-1}^A)^{\top}\bm x^A_{i}$ for every $t \ge 1$ and $i \in S_t$. In particular, when $d_A = 1$, he can learn the precise value of the data in terms of feature $x_A$ (up to a difference in sign) in at most $\max\{\lfloor\log_2(m)\rfloor + 1,2\}$ epochs with probability at least $1-\mathcal{O}(\frac{1}{m^{s-1}})$.
    \end{remark} 

    \paragraph{Large Random Initialization.} The strategy is to initialize each element of $\bm \theta_0$ to be $O(1)$, $e.g. \ $ $\mathcal{N}(0,1)$. This is less common in modern machine learning as it makes the output dependent on the dimension, $e.g.$  $\bm \theta_0^{\top}\bm x_i = O(d)$ in logistic regression. Although less favorable from the perspective of optimization, it is suitable for the vertical FL framework since now $\bm\theta_{t-1}^{\top}\bm x_i$ dominates the observable signal ${f}_{i,t}$, and therefore makes Criterion \ref{criterion} invalid. In some sense, one must \textit{trade efficiency for privacy}.

\medskip

\paragraph{Discussion of the Multi-party Scenario.}
{
Finally, we will generalize the label recovery attack to the multi-party scenario. The main observation is that the proposed attack only relies on the understanding of the \textit{full} inner product through the initialization scheme, and does not depend on the number of parties or how the features are split between them. }

{Specifically, suppose there are $p$ parties $U_1, \cdots, U_p$ with $p \ge 3$, $U_1$ is the only party that holds the label, and they own a disjoint set of the features $\bm x^{U_i}$ where $\bm x = \bm x^{U_1} \parallel ... \parallel \bm x^{U_p}$. Denote $d_i = \text{dim}(\bm x^{U_i}), \ i \in [p]$ and $d = \sum_{i=1}^p d_i$. Similar to the proof of Corollary \ref{corollary T}, we have the following theorem for Party $U_i,\  2 \le i \le p$. 
\begin{theorem}  \label{theorem multiple}
Suppose $s \le d_{U_i}$ for some $2 \le i \le p$. If $\bm \theta_0$ is initialized (with high probability) such that
\begin{align}
    \max_i \{|\bm \theta_0^{\top}\bm x_i|\} = \epsilon <2,
\end{align}
then Criterion \ref{criterion} is valid for $U_i$ for 
\begin{align}
    T_{\epsilon} = \bigg\lceil \log_{\left(1+\frac{\eta}{4}\right)}\left(\frac{4}{2+\epsilon}\right) \bigg\rceil
\end{align}
(with high probability).
\end{theorem}}

\section{Protocol-aware Active Attack}  \label{sec_discuss}

For the privacy analysis in Section \ref{sec_security_A} and \ref{sec_security_B}, we assume the existence of an oracle such that both parties can obtain their local gradients, and ignore the details of practical implementation. In this section, we will describe the protocol through which both parties can cooperatively learn their own share of gradient, without violating the privacy of individual instance \textit{directly}. Further, by modifying the protocol maliciously, both parties can relax the constraints of batch size as appears in the previous sections, and therefore are able to perform the same attack in more general scenarios. We refer to this as `\textit{active attack}'.

\subsection{Protocol of Obtaining Gradient}  \label{subsub protocol}
We will discuss how both parties collaboratively obtain $\nabla^A_{t}$, and the procedure for computing $\nabla^{B}_{t}$ shares a similar philosophy. Aside from the information known to himself, Party $A$ needs to compute $\bm H^A_{t}:=\sum_{i\in \mathcal{B}_t} g_{i,t}^B\bm x^A_{i}$ via the assistance from Party $B$, where $g_{i,t}^B =  (\bm \theta_{t-1}^B)^\top \bm x^B_{i}$. Also, Party $B$ cannot directly send the value of $g_{i,t}^B$ to Party $A$ as it will violate the privacy of individual $i$. 

Note $\bm H^A_{t}$ is essentially a linear combination of the $\bm x^A$ feature, with the coefficients being unknown. Therefore, a natural solution to the above problem is (additively) homomorphic encryption, such as Paillier \cite{paillier1999public} and OU \cite{okamoto1998new}. Here we will resort to the former one. In Paillier cryptosystem, there exists a pair of keys $pk_B$ and $sk_B$, where both parties can use $pk_B$ to encrypt a plaintext $a$ as its corresponding ciphertext $[\![a]\!]$, but only Party $B$ can use $sk_B$ to decrypt $[\![a]\!]$ and get $a$, which we denote as $D([\![a]\!]) = a$. For vector $\bm v$, we will use $[\![\bm v]\!]$ to denote element-wise encryption and similarly $\bm D$ for decryption. 

The most notable property of $[\![\cdot]\!]$ is additively homomorphic, i.e. $[\![a]\!]+[\![b]\!]=[\![a+b]\!]$ and $k[\![a]\!]=[\![ka]\!]$ for all plaintext $a,b$ and $k \in \mathbb{Z}$. We are now ready to present the protocol{\footnote{The protocol is standard and resembles a few works in the literature \cite{FATE-HeterLR,hardy2017private} --- the only difference here is that we do not assume the existence of a trusted third party $C$. Also, we assume all the quantities presented in this protocol appear in the form of integers and ignore the process of encoding and decoding. }}

\begin{algorithm}[htbp]
	\floatname{algorithm}{Protocol}
	\renewcommand{\algorithmicrequire}{\textbf{Input:}}
	\renewcommand{\algorithmicensure}{\textbf{Output:}}
	\caption{Party $A$ and $B$ compute the gradient $\nabla^A_{t}$}
	\label{alg1}
	\begin{algorithmic}[1]
		\REQUIRE The public key $pk_B$ and $d_A, d_B$ for both parties \\
		\hspace{6mm} The secret key $sk_B$ for $B$
		\STATE \textbf{Initialize} $\bm S^A_{t}\gets [\![\bm 0_{d_A}]\!]$
		\FOR{$i\in \mathcal{B}_t$}
		\STATE $B$ computes $g_{i,t}^B =  (\bm \theta_{t-1}^B)^\top \bm x^B_{i}$ and encrypts it as $[\![g_{i,t}^B]\!]$
		\ENDFOR
		\STATE $B$ sends $[\![g_{i,t}^B]\!]$, $i\in \mathcal{B}_t$ to $A$
		\FOR{$i\in \mathcal{B}_t$}
		\STATE $A$ computes $\bm S^A_{t}\gets \bm S^A_{t}+[\![g_{i,t}^B]\!]\bm x^A_{i}$  \label{steal}
		\ENDFOR
		\STATE $A$ generates a random vector $\bm r \in \mathbb{Z}^{d_A}$ and computes $[\![\bm r]\!]$
		\STATE $A$ masks the summation $\widetilde{\bm S}^A_t := \bm S^A_{t}-[\![\bm r]\!]$ and sends it to $B$
		\STATE $B$ decrypts to get $\bm D(\widetilde{\bm S}^A_t)$ and sends it back to $A$
		\STATE $A$ computes $\bm G^A_{t} := \bm D(\widetilde{\bm S}^A_t)+\sum_{i\in \mathcal{B}_t}(g_{i,t}^A-2y_i)\bm x^A_{i}+ \bm r$, where $g_{i,t}^A=(\bm \theta_{t-1}^A)^\top \bm x^A_{i}$  
		\ENSURE The gradient for $A$ at iteration $t$: $\nabla^A_{t}=\frac{1}{4s}\bm G^A_{t}$ 
	\end{algorithmic}  
\end{algorithm}

To see the correctness of Protocol \ref{alg1}, note Party $B$ receives
\begin{align}
    \widetilde{\bm S}^A_t &= [\![\bm 0_{d_A}]\!] + \sum_{i \in \mathcal{B}_t} [\![g_{i,t}^B]\!]\bm x^A_{i} - [\![\bm r]\!] \\
    &= \bigg[\!\!\bigg[\sum_{i \in \mathcal{B}_t}g_{i,t}^B\bm x^A_{i} - \bm r\bigg]\!\!\bigg] = [\![\bm H^A_{t} - \bm r]\!],  \label{Eq. homo}
\end{align}
where we use additive homomorphic in Eq. (\ref{Eq. homo}). Therefore, Party $A$ can recover $\bm H^A_{t}$ by adding $\bm r$ to the plaintext $\bm D(\widetilde{\bm S}^A_t)$.

\subsection{Active Attack}  \label{subsub AA}

While constructing attacks in Section \ref{sec_security_A} and \ref{sec_security_B}, we assume the batch size is no more than the dimension of Party $A$'s (or Party $B$'s) features, i.e. $s\leq d_A$ (or $s\leq d_B$). Unfortunately, such attacks could be easily invalidated by increasing the batch size. To further strengthen our results, we propose generating auxiliary ciphertext during the execution of information transfer. Similarly, we will use Party $A$ as an example to illustrate our strategy and the results can be easily generalized to Party $B$.

The idea is that Party $A$ receives the encrypted coefficient $[\![g_{i,t}^B]\!]$ from Party $B$ and uses it to multiply $\bm x^A_{i}$ in line \ref{steal} of Protocol \ref{alg1}, so intuitively he can gain more information about $g_{i,t}^B$ by increasing the dimension of $\bm x_{A}$. To this end, Party $A$ extends his share of feature on the training samples to $\bm z^A_{i}:= \left(\begin{array}{@{}c@{}}
     \bm x^A_{i}  \\
     \hline
     \widetilde{\bm x}^A_{i}
\end{array}\right)
$ for $i \in [n]$, where $\{\widetilde{\bm x}^A_{i}\}_{i\in[n]}$ is randomly sampled from a continuous probability distribution and $\dim(\widetilde{\bm x}^{A}) = \widetilde{d}_A$. We also assume $\widetilde{d}_A = (u-1)d_A$ for some $u \in \mathbb{Z}^+$ and therefore $\dim(\bm z^{A}) = ud_A$.

Now, we modify line \ref{steal} of Protocol \ref{alg1} by multiplying each $[\![g_{i,t}^B]\!]$ with $\bm z^A_{i}$, and Party $A$ obtains
\begin{align}
    \bm E^A_{t}:=\bigg[\!\!\bigg[\sum_{i\in\mathcal{B}_t}g_{i,t}^B\bm z^A_{i}\bigg]\!\!\bigg] 
\end{align}
after the loop. But now $\text{dim}(\bm E^A_{t}) = ud_A$, so it will not be accepted by Party $B$. Therefore, we compress multiple ciphertexts into one ciphertext with longer plaintext length. Specifically, we use $\bm v[j]$ to denote the $j$-th element of vector $\bm v$, and assume each plaintext has been encoded as an $\iota$-bit integer. Now define
\begin{align}
    \bm S^A_{t}[i] := \sum_{s=0}^{u-1} \bm E^A_{t}[i+sd_A] * 2^{s\iota}, \ \ \ \forall i \in [d_A],   \label{eq compress}
\end{align}
then the dimension of $\bm D\left(\bm S^A_{t}\right)$ will be $d_A$. Further, for each $i \in [d_A]$, $ D\left(\bm S^A_{t}[i]\right)$ can be \textit{uniquely} split into $u$ integers by $\iota$-bit since 
\begin{align}
    D\left(\bm S^A_{t}[i]\right) = \sum_{s=0}^{u-1} D\left(\bm E^A_{t}[i+sd_A]\right) * 2^{s\iota}.
\end{align}
Therefore, $\bm D(\bm E^A_{t})$ can be recovered from $\bm D(\bm S^A_{t})$. Finally, before sending $\bm S^A_{t}$ to Party $B$ for decryption, Party $A$ will similarly generate a noise vector $\bm r$ and mask the object. This step ensures that Party $B$ cannot distinguish whether Party $A$ has hidden any information in $\bm S^A_{t}$. We now present the complete algorithm as follow.

\begin{algorithm}[htbp]
	\floatname{algorithm}{Algorithm}
	\renewcommand{\algorithmicrequire}{\textbf{Input:}}
	\renewcommand{\algorithmicensure}{\textbf{Output:}}
	\caption{Active attack of Party $A$}
	\label{alg2}
	\begin{algorithmic}[1]
		\REQUIRE $[\![g_{i,t}^B]\!], i \in \mathcal{B}_t$ for $A$
		\STATE \textbf{Initialize} $\bm E^A_{t}\gets [\![\bm 0_{ud_A}]\!]$
		\FOR{$i\in \mathcal{B}_t$}
		\STATE $A$ computes $\bm E^A_{t}\gets \bm E^A_{t}+[\![g_{i,t}^B]\!]\bm z_{A,i}$ 
		\ENDFOR
		\STATE $A$ compresses $\bm E^A_{t}$ into $\bm S^A_{t}$ via Eq. (\ref{eq compress})
		\STATE $A$ generates a random vector $\bm r \in \mathbb{Z}^{{d_A}}$ and computes $[\![\bm r]\!]$
		\STATE $A$ masks the object $\widetilde{\bm S}^A_{t} := \bm S^A_{t}-[\![\bm r]\!]$ and sends it to $B$
		\STATE $B$ decrypts to get $\bm D(\widetilde{\bm S}^A_{t})$ and sends it back to $A$
		\STATE $A$ computes $\bm D(\widetilde{\bm S}^A_{t}) + \bm r$, then splits it by $\iota$-bit

		\ENSURE The plaintext of $\bm E^A_{t}$: $\bm D(\bm E^A_{t})$ 
	\end{algorithmic}  
\end{algorithm}

Note that the output of Algorithm \ref{alg2} essentially provides Party $A$ a linear system
\begin{align}
    \sum_{i\in\mathcal{B}_t}g_{i,t}^B\bm z^A_{i} = \bm D(\bm E^A_{t}),  \label{eq leakage}
\end{align}
which has $s$ variables $g_{i,t}^B, \ i \in \mathcal{B}_t$ and $ud_A$ linear equations. Compared to the original protocol which gives $d_A$ linear equations, Party $A$ obtains $(u-1)d_A$ additional equations via the generation of auxiliary ciphertext. Further, since $\{\widetilde{\bm x}_{A,i}\}_{i\in[n]}$ is sampled from a continuous probability distribution, we can apply Theorem \ref{Theorem Linear} and conclude the uniqueness of $g_{i,t}^B, \ i \in \mathcal{B}_t$ when $s \le ud_A$. This further allows us to perform the same attack as in Section \ref{sec_security_A}. Analogously, Party $B$ can adopt the same strategy and perform the same attack in Section \ref{sec_security_B} when $s \le ud_B$.

Finally, we will pick a suitable $u$ so that the Paillier cryptosystem can work properly. Denote $N$ as the maximum number that we can apply encryption on and the length of $N$ is $\kappa$-bit. Note in Eq. (\ref{eq compress}) we produce a ciphertext whose plaintext has length at most $u\iota$. Therefore, to ensure the correctness of encryption, we can pick $u =\lfloor\frac{{\kappa-1}}{\iota}\rfloor$ and Algorithm \ref{alg2} can still function properly. In practice, a $\iota=64$-bit encoding scheme on the plaintext can guarantee good model performance
and the length of {$N$} is $\kappa=2048$-bit or even more, implying $u \geq {31}$. {Therefore, unless $s > 31\max\{d_A,d_B\}$, at least one of the attacks in Section \ref{sec_security_A} or \ref{sec_security_B} will still be effective and therefore can be applied to more general scenarios in practice.} This is a significant improvement over the original protocol. {Additionally, if $d_A$ and $d_B$ is not publicly available, then both party can maliciously claim a larger `$d_A$' or `$d_B$' to obtain more linear equations from the other side. The constraints on batch size will then be negligible.}

In all, the proposed active attack is mainly based on \textit{generating} and \textit{compressing} auxiliary ciphertext. While the idea of creating linear equations is straightforward and does appear in a previous work \cite{li2021privacy}, in context of the HE-based protocol where the linear equations essentially come from the decoding party, we cannot simply attach the auxilliary ciphertext to the original ones, which would result in denial of decoding. Instead, we need to project the high-dimensional vector into the original domain to permit decoding in a \textit{lossless} manner. This is guaranteed by our compressing mechanism, which leverages the power-of-$2$ summation to ensure unique decodability. Together with the capacity of the Paillier cryptosystem as well as the encoding scheme on the plaintext, we can obtain a lower bound on the number of linear equations guaranteed by our active attack.

\section{Countermeasure}  \label{sec_defense}

So far we've been mostly discussing privacy attacks in vertical logistic regression. In this section, we will complement the previous results by proposing a countermeasure. Note merely increasing the batch size might not be adequate due to the active attack demonstrated in Subsection \ref{subsub AA}.

Take Party $A$ as an example. It is straightforward to see from Algorithm \ref{alg2} that the privacy breach arises from the linear system Eq. (\ref{eq leakage}), through which Party $A$ can learn the precise value of $g_{i,t}^B$. To fix this issue, a natural approach\footnote{Aside from the current framework, another straightforward solution is to use secret sharing \cite{beimel2011secret}, so that both parties can only obtain random shares of the gradient throughout the training process. See \cite{mohassel2017secureml} for an example.} for Party $B$ is to modify Protocol \ref{alg1} and mask the plaintext with random noise before encryption, so that the solution to the linear system deviates from the true value $g_{i,t}^B$. This strategy corresponds to \textit{differential privacy} \cite{dwork2006differential}, which is a classic tool in secure multi-party computation and machine learning. 

\subsection{DP-Paillier-MGD}


Let us begin with the formal definition of Differential Privacy (DP). 

\begin{definition}[Differential Privacy]

A randomized algorithm $\mathcal{A}$ with domain $\mathcal{D}$ and range $\mathcal{R}$ is $(\epsilon, \delta)$-differentially private if for all $\mathcal{S} \in \mathcal{R}$ and
for all $D_1, D_2 \in \mathcal{D}$ that differ by at most one instance,
\begin{align*}
    \Pr(\mathcal{A}(D_1)\in \mathcal{S}) \le \exp{(\epsilon)}\Pr(\mathcal{A}(D_2)\in \mathcal{S}) + \delta.
\end{align*}
\end{definition}

DP is usually realized via \textit{randomization}: making the output undeterministic by adding noise. Here we set the noise vector to be Gaussian, which has been widely adopted in practice \cite{abadi2016deep, bassily2014private}.

\begin{definition}[Gaussian Mechanism]
Given $f: \mathcal{D} \to \mathbb{R}^d$, the Gaussian Mechanism  with parameter $\sigma$ is defined as:
\begin{align}
    M_G\left(D, f(\cdot), \sigma\right) := f(D) + (Y_1, \cdots, Y_d)^{\top},
\end{align}
where $D \in \mathcal{D}$ and $Y_i$ are i.i.d. random variables drawn from $\mathcal{N}(0, \sigma^2)$.
\end{definition}

{An important caveat here is that we will not resort to the original definition of DP --- which is interpreted as a privacy guarantee regarding the \textit{membership inference attack} --- because Party $A$ and $B$ need to align their user ID before initializing the procedure of VFL and hence hold the data records in the training set. Instead, we seek to leverage the effect of \textit{obfuscation} brought along by the noise-adding mechanism, to ensure that DP also serves as a good countermeasure for the attacks proposed in the previous sections.  
}

Specifically, the strategy for Party $B$ is to mask the sensitive information, and therefore Party $A$ cannot learn the \textit{precise} value of $g_{i,t}^B$ after solving the linear system. Similarly, Party $A$ will also mask the sensitive information $g_{i,t}^A - 2y_i$ before encryption. Denote $\bm{SV}^B_t$ as the sensitive vector concatenated by all $\{g_{i,t}^B\}_{i \in \mathcal{B}_t}$ and similarly define $\bm{SV}^A_t$. Both parties adopt Gaussian Mechanism for masking, and the updated algorithm, which we name as DP-Paillier-MGD, is given as follow.

\begin{algorithm}[htbp]
	\floatname{algorithm}{Algorithm}
	\renewcommand{\algorithmicrequire}{\textbf{Input:}}
	\renewcommand{\algorithmicensure}{\textbf{Output:}}
	\caption{DP-Paillier-MGD}
	\label{alg3}
	\begin{algorithmic}[1]
		\REQUIRE mini-batch $\mathcal{B}_t$, learning rate $\eta_t = \eta$, two Paillier cryptosystems $\mathscr{P}_1, \mathscr{P}_2$  \\
		\hspace{6mm} variance $\sigma_A^2, \sigma_B^2$ of noise vector for $A, B$
		\STATE \textbf{Initialize} $\bm \theta_0 = \bm \theta_0^A \parallel \bm \theta_0^B$
		\FOR{$t = 1: T$}
		\STATE $B$ generates $\bm Z^B \sim \mathcal{N}(\bm 0, \sigma_B^2\bm I_s)$, computes $Sec(\bm{SV}^B_t) := \bm{SV}^B_t + \bm Z^B$ and encrypts it in $\mathscr{P}_1$
		
		\STATE both parties follow the procedure of Protocol \ref{alg1} using $\mathcal{P}_1$ until $A$ obtains $\bm X_{\mathcal{B}_t}^A Sec(\bm{SV}_t^B)$, where $\bm X_{\mathcal{B}_t}^A$ is formed via stacking $\{\bm x_i^A\}_{i\in\mathcal{B}_t}$ by column 
		\STATE $A$ computes $\widetilde{\bm G}_t^A := \bm X_{\mathcal{B}_t}^A\bm{SV}_{t}^A+ \bm X_{\mathcal{B}_t}^ASec(\bm{SV}_t^B)$
		
		\STATE $A$ generates $\bm Z^A \sim \mathcal{N}(\bm 0, \sigma_A^2\bm I_s)$, computes $Sec(\bm{SV}^A_t) := \bm{SV}^A_t + \bm Z^A$ and encrypts it in $\mathscr{P}_2$
		
		\STATE both parties follow the procedure of Protocol \ref{alg1} using $\mathcal{P}_2$ until $B$ obtains $\bm X_{\mathcal{B}_t}^B Sec(\bm{SV}_{t}^A))$, where $\bm X_{\mathcal{B}_t}^B$ is similarly defined as $\bm X_{\mathcal{B}_t}^A$
		\STATE $B$ computes $\widetilde{\bm G}_t^B := \bm X_{\mathcal{B}_t}^B\bm{SV}_{t}^B+ \bm X_{\mathcal{B}_t}^B Sec(\bm{SV}_{t}^A)$

		\STATE $A$ updates $\bm \theta_t^A = \bm \theta_{t-1}^A - \frac{\eta}{4s}\widetilde{\bm G}_t^A$, $B$ updates $\bm \theta_t^B = \bm \theta_{t-1}^B - \frac{\eta}{4s}\widetilde{\bm G}_t^B$
		\ENDFOR

		\ENSURE $\{Sec(\bm{SV}_t^A)\}_{t \in [T]}, \{Sec(\bm{SV}_t^B)\}_{t \in [T]}$
	\end{algorithmic}  
\end{algorithm}

\subsection{Utility and Privacy Guarantees} \label{subsec PU}
In this subsection, we will present utility and {(worst-case)} privacy guarantees for Algorithm \ref{alg3}. We assume 
\begin{align}
    \max_{t \in [T]} \|\bm\theta_{t-1}\|_2 \le 4G-2
\end{align}
for some $G > \frac{1}{2}$, so that
{\begin{align}
    \max\limits_{i \in [n], t \in [T]}\|\nabla\mathcal{L}_{i}(\bm\theta_{t-1})\|_2 &\le \frac{\max\limits_{i \in [n], t \in [T]} |\bm\theta_{t-1}^{\top}\bm x_i - 2y_i|}{4} \le G,
\end{align}}
where $\nabla\mathcal{L}_{i}(\bm\theta_{t-1})$ is the gradient on instance $i$.

\subsubsection{Utility Analysis}
We begin ourselves on the utility side and provide a convergence analysis for Algorithm \ref{alg3} under the noise level $\sigma_A$ and $\sigma_B$. Note that under the first-order approximation scheme, the original update on $\bm \theta$ can be treated as mini-batch gradient descent of the following loss function:
\begin{align}
    \mathcal{L}(\bm \theta) := \frac{1}{8n}\sum_{i=1}^n (\bm \theta^{\top}\bm x_i - 2y_i)^2,    \label{Eq approx loss}
\end{align}
with the mini-batch gradient $\nabla\mathcal{L}_{\mathcal{B}_t}(\bm\theta_{t-1})$ identical to the one defined in Eq. (\ref{eq MGD-taylor}). After injecting noise in Algorithm \ref{alg3}, the update on $\bm \theta$ becomes
\begin{align}
    \bm \theta_t = \bm \theta_{t-1} - \eta_t \left(\nabla\mathcal{L}_{\mathcal{B}_t}(\bm\theta_{t-1}) + {\bm e_t}\right),   \label{defense SGD}
\end{align}
where the gradient error is
\begin{align}
    \bm e_t := \frac{1}{4s}\left(\begin{array}{c}
         \bm X_{\mathcal{B}_t}^A\bm Z^B  \\
         \bm X_{\mathcal{B}_t}^B\bm Z^A 
    \end{array}\right).
\end{align}

Our first step is to control the magnitude of $\bm e_t$, which is mainly done by a standard concentration inequality and deferred to Appendix \ref{app_sub_C}. 
\begin{lemma} \label{lemma gradient error}
Denote $\sigma := \max\{\sigma_A, \sigma_B\}$. Given $k \ge 1$, we have $\|\bm e_t\|_2 \le \sqrt{\frac{k}{s}}\sigma$ for $t \in [T]$ with probability at least $1 - 2T\exp(-ck)$,  where $c$ is an absolute constant.  
\end{lemma}

With Lemma \ref{lemma gradient error} at hand, we are now ready to state the utility guarantee of Algorithm \ref{alg3}. It is basically a combination of the standard mirror descent argument, as well as a martingale trick to control the stacking of noise over the training horizon. 

\begin{theorem}   \label{theorem utility}
Let $\bm \theta^*$ be an arbitrary reference point. Assume 
\begin{align}
    \max\limits_{t \in [T]} \|\bm \theta_{t-1} - \bm \theta^*\|_2 \le C,
\end{align}
then with step size $\eta = {\mathcal{O}}(\frac{1}{\sqrt{T\ln T}})$, Algorithm \ref{alg3} satisfies
\begin{align}
    \mathcal{L}\left(\frac{1}{T}\sum_{t \in [T]} \bm \theta_{t-1} \right) \le \mathcal{L}(\bm \theta^*) + \mathcal{O}\left(\frac{u+1}{\sqrt{T}}\right)
\end{align}
with probability at least $1-2\delta$, where $u := \sqrt{\frac{\ln(2T/\delta)}{cs}}\sigma$.
\end{theorem}

\begin{proof}
Let $\delta = 2Te^{-ck}$ in Lemma \ref{lemma gradient error}, we have with failure probability at most $\delta$, 
\begin{align}
    \|\bm e_t\|_2 \le \sqrt{\frac{k}{s}}\sigma = \sqrt{\frac{\ln(2T/\delta)}{cs}} = u
\end{align}
for $t\in [T]$. Now denote
\begin{align}
    \bm g_t := \nabla\mathcal{L}_{\mathcal{B}_t}(\bm\theta_{t-1}) \ \ \ \text{and} \ \ \ \bm G_t := \nabla\mathcal{L}(\bm\theta_{t-1}).
\end{align}
Since $\|\nabla\mathcal{L}_{i}(\bm\theta_{t-1})\|_2 \le G$, we have $\|\bm g_t\|_2, \|\bm G_t\|_2 \le G$ by triangle inequality. Now for $t \in [T]$, consider
\begin{align}
    & \ \ \ \ \|\bm \theta_t - \bm \theta^*\|_2^2 \\
    &=\|\bm \theta_{t-1}- \bm \theta^*\|_2^2  + 2\eta\langle \bm \theta^*- \bm \theta_{t-1}, \bm g_t + \bm e_t - \bm G_t  \rangle \nonumber \\
    & \ \ \ + 2\eta\langle \bm \theta^*- \bm \theta_{t-1}, \bm G_t  \rangle + \eta^2\|\bm g_t + \bm e_t\|_2^2 \\
    &\le \|\bm \theta_{t-1}- \bm \theta^*\|_2^2 + 2\eta\underbrace{\langle \bm \theta^*- \bm \theta_{t-1}, \bm g_t + \bm e_t - \bm G_t  \rangle}_{\epsilon_t} \nonumber \\ 
    & \ \ \ + 2\eta(\mathcal{L}(\bm \theta^*)-\mathcal{L}(\bm \theta_{t-1})) + 2\eta^2\left(G + \|\bm e_t\|_2\right)^2,
\end{align}
where we use the convexity of $\mathcal{L}$ in the first inequality. Rearranging both side, we have
\begin{align}
    & \ \ \ \ \mathcal{L}(\bm \theta_{t-1})  \\
    &\le \mathcal{L}(\bm \theta^*) + \frac{\|\bm \theta_{t-1}- \bm \theta^*\|_2^2 - \|\bm \theta_{t}- \bm \theta^*\|_2^2}{2\eta} +  \frac{\eta}{2}\left(G + u\right)^2 + \epsilon_t.
\end{align}
Now apply $\frac{1}{T}\sum\limits_{t \in [T]}$ to both sides, and note the convexity of $\mathcal{L}$ again, we have
\begin{align}
    & \ \ \ \ \mathcal{L}\left(\frac{1}{T}\sum_{t \in [T]} \bm \theta_{t-1} \right) \\ 
    &\le \frac{1}{T}\sum\limits_{t \in [T]} \mathcal{L}(\bm \theta_{t-1}) \\
    &\le \mathcal{L}(\bm \theta^*)+ \frac{\|\bm \theta_{0}- \bm \theta^*\|_2^2}{2\eta T} + \frac{\eta}{2}\left(G + u\right)^2 + \frac{1}{T}\sum_{t \in [T]} \epsilon_t.
\end{align}
Setting $\eta = \frac{\|\bm \theta_{0}- \bm \theta^*\|_2}{(G+u)\sqrt{T}}$, we have
\begin{align}
    \mathcal{L}\left(\frac{1}{T}\sum_{t \in [T]} \bm \theta_{t-1} \right) &\le \mathcal{L}(\bm \theta^*) + \frac{(G+u)C}{\sqrt{T}} + \frac{1}{T}\sum_{t \in [T]} \epsilon_t.  \label{eq utility bound}
\end{align}
Finally, it suffices to control the sum of error term $\sum_{t \in [T]} \epsilon_t$. Note
\begin{align}
    \mathbb{E}[\epsilon_t \mid \bm \theta_{< t}] &= \mathbb{E}[\langle \bm \theta^*- \bm \theta_{t-1}, \bm g_t + \bm e_t - \bm G_t  \rangle \mid \bm \theta_{< t}]  \\
    &= \big<\mathbb{E}[\bm g_t + \bm e_t - \bm G_t \mid \bm \theta_{< t}],\bm \theta^*- \bm \theta_{t-1}\big> \\
    &=0,
\end{align}
so $\{\epsilon_t\}_{t \in [T]}$ is a  martingale difference sequence. Further, by Cauchy-Schwarz and triangle inequality we have
\begin{align}
    \mathbb{E} |\epsilon_t| &\le \|\bm \theta^* - \bm \theta_{t-1}\|_2 (\|\bm g_t\|_2+\|\bm G_t\|_2 + \|\bm e_t\|_2) \\
    &\le (2G+u)C.
\end{align}
Therefore, by Azuma-Hoeffding inequality (Corollary 2.20 in \cite{wainwright2019high}) we have
\begin{align}
    \sum_{t \in [T]} \epsilon_t \le (2G+u)C\sqrt{2T\ln(1/\delta)}  \label{eq azuma}
\end{align}
with probability at least $1-\delta$. Plugging Eq. (\ref{eq azuma}) into Eq. (\ref{eq utility bound}) gives the result as desired.
\end{proof}

\begin{remark}
The error in the right hand side consists of two parts: the first term is $\mathcal{L}(\bm \theta^*)$, which can be made small by considering an optimal reference point; the second term is $\widetilde{\mathcal{O}}(\frac{1}{\sqrt{T}})$, which matches the standard convergence rate of mini-batch gradient descent in the convex setting \emph{\cite{duchi2018introductory}} up to a logarithmic factor, and achieves the exact rate when there is no noise, i.e. $\sigma = 0$. 
\end{remark}

\subsubsection{Privacy Analysis}
{
We will now turn to the privacy side. As mentioned at the beginning of this section, we will not argue that Algorithm \ref{alg3} is $(\epsilon, \delta)$-differentially private, mainly because this is usually seen as a privacy guarantee regarding membership inference attack. Therefore, we choose the upper bound on $T^*$ that we can safely apply Criterion \ref{criterion} as a \textit{measurement} of privacy leakage. We focus on the label recovery attack as the feature attack is expected to fail, since it relies on the precise value of the sensitive information and therefore can be easily broken by noise, whereas the label recovery attack only requires the sensitive information to fall in a certain range. Note such upper bound might not accurately reflect practice as we are actually considering the \textit{worst case} scenario and there could be cancellation of noise in reality; nevertheless, compared to the similar result presented in Corollary \ref{corollary T}, this upper bound still sheds light on the effectiveness of Gaussian mechanism. We will use experiments to provide stronger support in the next section.}

\begin{theorem}  \label{theorem defense ub}
Suppose $\max\limits_{i} \{|\bm\theta_0^{\top}\bm x_i|\} = \epsilon < 2$, then
Criterion \ref{criterion} is valid for Algorithm \ref{alg3} with
\begin{align}
    \widetilde{T_{\epsilon}} = \bigg\lceil \log_{\left(1+\frac{\eta}{4}\right)}\left(\frac{4u+4}{4u+2+\epsilon}\right) \bigg\rceil < {T_{\epsilon}}
\end{align}
with probability at least $1 - \delta$, where $u$ is defined in Theorem \ref{theorem utility}.
\end{theorem}

\begin{proof}
Define $q_t = \max\limits_{i} \{|\bm \theta_{t-1}^{\top}\bm x_i| + p \}$ where $p = 4u+2$. By the update formula of Algorithm \ref{alg3}, which is given by Eq. (\ref{defense SGD}), for $1\leq j \leq n$ we have
\begin{align}
    \bm\theta_t^{\top}\bm x_j = \bm \theta_{t-1}^{\top}\bm x_j - \frac{\eta}{4s} \sum_{i \in \mathcal{B}_t} (\bm \theta_{t-1}^{\top}\bm x_i - 2y_i)\bm x_i^{\top}\bm x_j - \eta \bm e_t^{\top}\bm x_j.
\end{align}
Discarding the failure probability and applying Cauchy-Schwarz and triangle inequality, we have
\begin{align}
    \big|\bm \theta_t^{\top}\bm x_j\big| \le \big|\bm \theta_{t-1}^{\top}\bm x_j\big| + \frac{\eta}{4s}\sum_{i \in \mathcal{B}_t}\left(\big|\bm \theta_{t-1}^{\top}\bm x_i\big|+2\right) + \eta u,
\end{align}
with probability at least $1-\delta$, implying
\begin{align}
    \left(\big|\bm\theta_t^{\top}\bm x_j\big|+p\right) \le \left(\big|\bm\theta_{t-1}^{\top}\bm x_j\big|+p\right) + \frac{\eta}{4s}\sum_{i \in \mathcal{B}_t}\left(\big|\bm\theta_{t-1}^{\top}\bm x_i\big|+p\right),
\end{align}
Taking maximum over $1\leq j \leq n$ on both side, we immediately have
\begin{align}
    q_{t+1} \le \left(1+\frac{\eta}{4}\right) q_{t}.
\end{align}
Finally, note $q_1 = 4u + 2 + \epsilon$ and the threshold of failure is $4u + 4 + \epsilon$. Therefore, Criterion \ref{criterion} is valid for
\begin{align}
    \widetilde{T_\epsilon} = \bigg\lceil \log_{\left(1+\frac{\eta}{4}\right)}\left(\frac{4u+4}{4u+2+\epsilon}\right) \bigg\rceil.
\end{align}
\end{proof}

\section{Experiments} \label{sec_experiment}

In this section, we implement vertical logistic regression on practical problems including classifying digits 0 and 1 in MNIST dataset \cite{lecun2010mnist} and predicting default payments on UCL default of credit card clients Data Set \cite{yeh2009comparisons} to showcase the effectiveness of our label attack and defense strategy.

\paragraph{Data.} MNIST is a dataset of the images of handwritten digits from 0 to 9. For evaluating vertical logistic regression, we take all the 0, 1 from MNIST and construct a binary classification problem. We assign label -1 to digit 0 and label 1 to digit 1 in the binary classification problem. There are 5923 training samples and 980 test samples for digit 0; 6742 training samples and 1135 test samples for digit 1, respectively. UCL default of credit card clients Data Set (abbreviated as credit) contains information on default payments, demographic factors, credit data, history of payment, and bill statements of credit card clients in Taiwan from April 2005 to September 2005. There are 30000 samples and we randomly split them into a training set of 27000 samples and a test set of 3000 samples. Each sample has a 23-dimensional feature and a binary label. We standardize the features of the credit dataset before processing logistic regression.

\paragraph{Learning Setting.} For both datasets, we consider the vertical logistic regression where Party A can access the first half of the features and Party B can access the last half. The batch size is set to 256. {Note here we focus on the validity of Criterion \ref{criterion}, which only involves the value of the full inner product, hence independent of the choices of batch size or the splitting ratio.} We mildly tune the learning rate from $[0.001, 0.003, 0.01]$ and choose 0.01 as the final learning rate. We evaluate 3 commonly-used initialization schemes including zero initialization (set all parameters to 0), Xavier Initialization~\cite{glorot2010understanding} Kaiming Initialization~\cite{he2015delving}. We set the total number of epochs to 30 and 10 for MNIST and credit dataset, respectively. For each setting, we repeat the experiment with 5 different random seeds and plot both mean (solid curve) and standard error (shadow).

\paragraph{Attack and Defense.} We evaluate our attack by calculating the alignment of the label obtained from Criterion~\ref{criterion} with the true label. For defense, we implement Algorithm~\ref{alg3} with different $\epsilon\in[0.1, 0.2, 0.5]$ and fix $\delta$ to 0.1. {Although we don't consider $(\epsilon,\delta)$-DP as a \textit{formal} privacy guarantee, we can still heuristically estimate the $\ell_2$-sensitivity iteratively and demonstrate that Algorithm \ref{alg3} satisfies $(\epsilon, \delta)$-DP with}
\begin{align}  \label{DP_A}
    \sigma_A = \sqrt{2\ln\left(\frac{5}{4\delta}\right)}\frac{\sqrt{\frac{8G^2e^2T\eta^2}{s} + 64G^2e}}{\epsilon},
\end{align}
\begin{align}  \label{DP_B}
    \sigma_B = \sqrt{2\ln\left(\frac{5}{4\delta}\right)}\frac{\sqrt{\frac{8G^2e^2T\eta^2}{s} + (8G-4)^2e}}{\epsilon}.
\end{align}
The detailed derivation is deferred to Appendix \ref{app.exp}. We then calculate the standard deviations $\sigma_A,\sigma_B$ according to the above equations where $T$ is chosen to be $30$ and $G$ is chosen to be $1$.

\paragraph{Results on MNIST.} In Figure~\ref{fig:exp_acc},~\ref{fig:exp_stats}, we report the experiment results on MNIST under the setting where the initialization follows Xavier Initialization. We observe that the performance of VFL under our defense is comparative to that without defense, especially for larger $\epsilon$. Meanwhile, the success rate of the label attack is significantly reduced. Though the theory of differential privacy encourages the choice of smaller $\epsilon$, we find that, in practice, even for $\epsilon$ as large as $0.5$, the success rate of label attack is as small as a random guess. This is because $\sigma_A, \sigma_B$ are already very large (larger than 10) for $\epsilon=0.5$. The results under Kaiming Initialization are similar. For zero initialization, the success rate of label attack stays $100\%$ for 10 more epochs than Xavier Initialization while the accuracy increases faster.

\begin{figure}[h]
\includegraphics[width=0.5\textwidth]{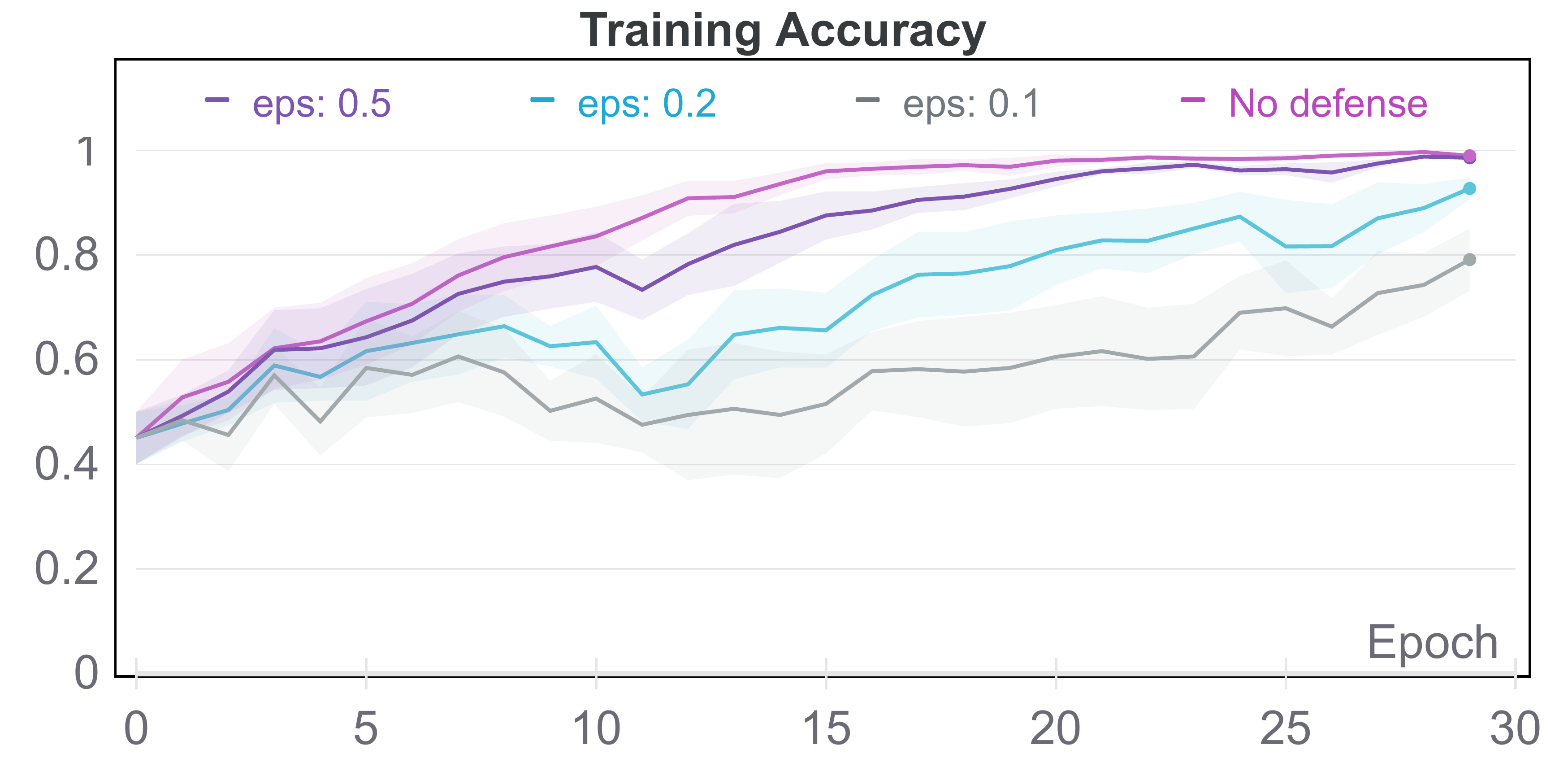}
\includegraphics[width=0.5\textwidth]{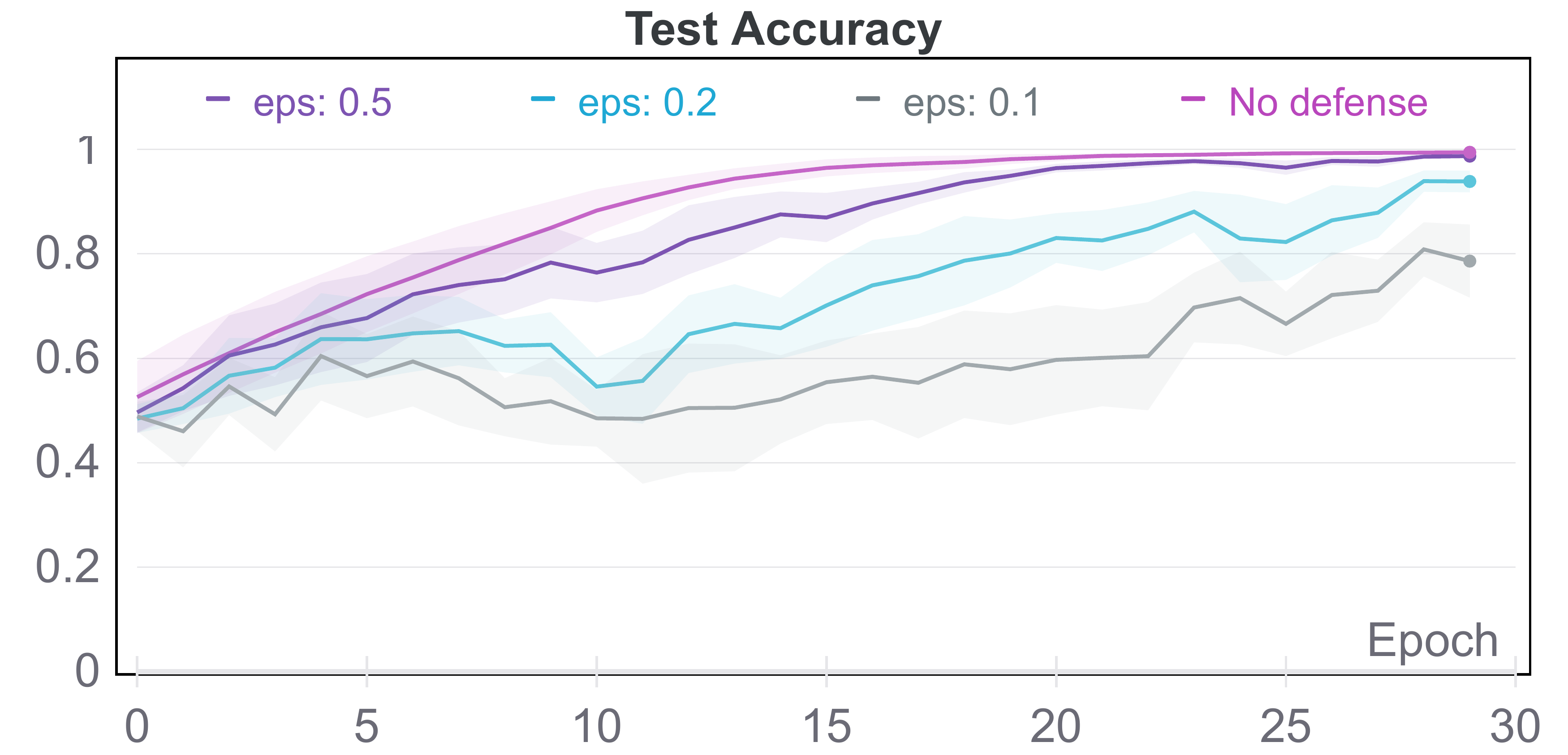}
\caption{Training and test accuracy under different settings on MNIST. VFL can achieve good performance, i.e., more than $90\%$ test accuracy, even for $\epsilon$ as small as $0.2$.}
\label{fig:exp_acc}
\end{figure}
\begin{figure}[h]
\includegraphics[width=0.5\textwidth]{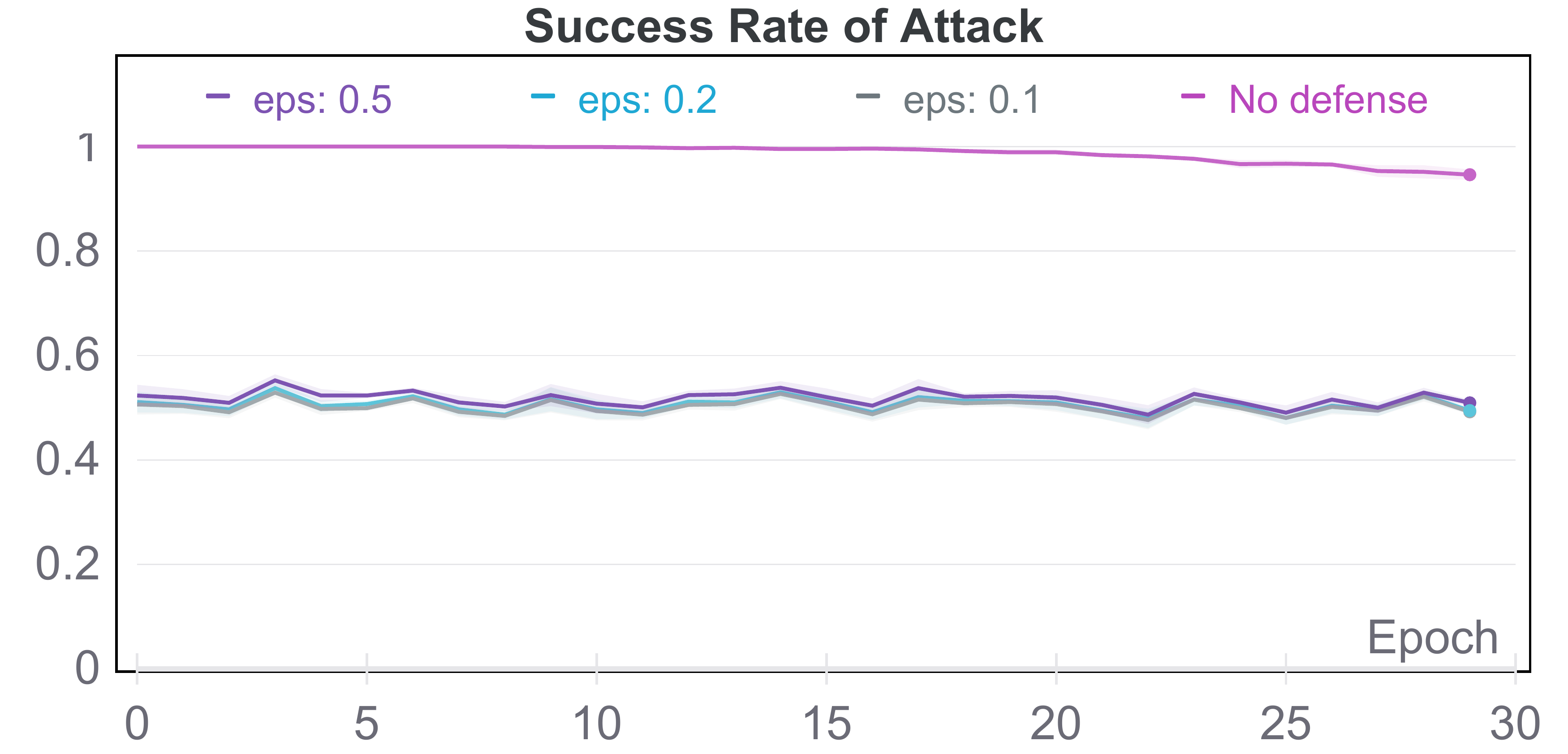}
\includegraphics[width=0.5\textwidth]{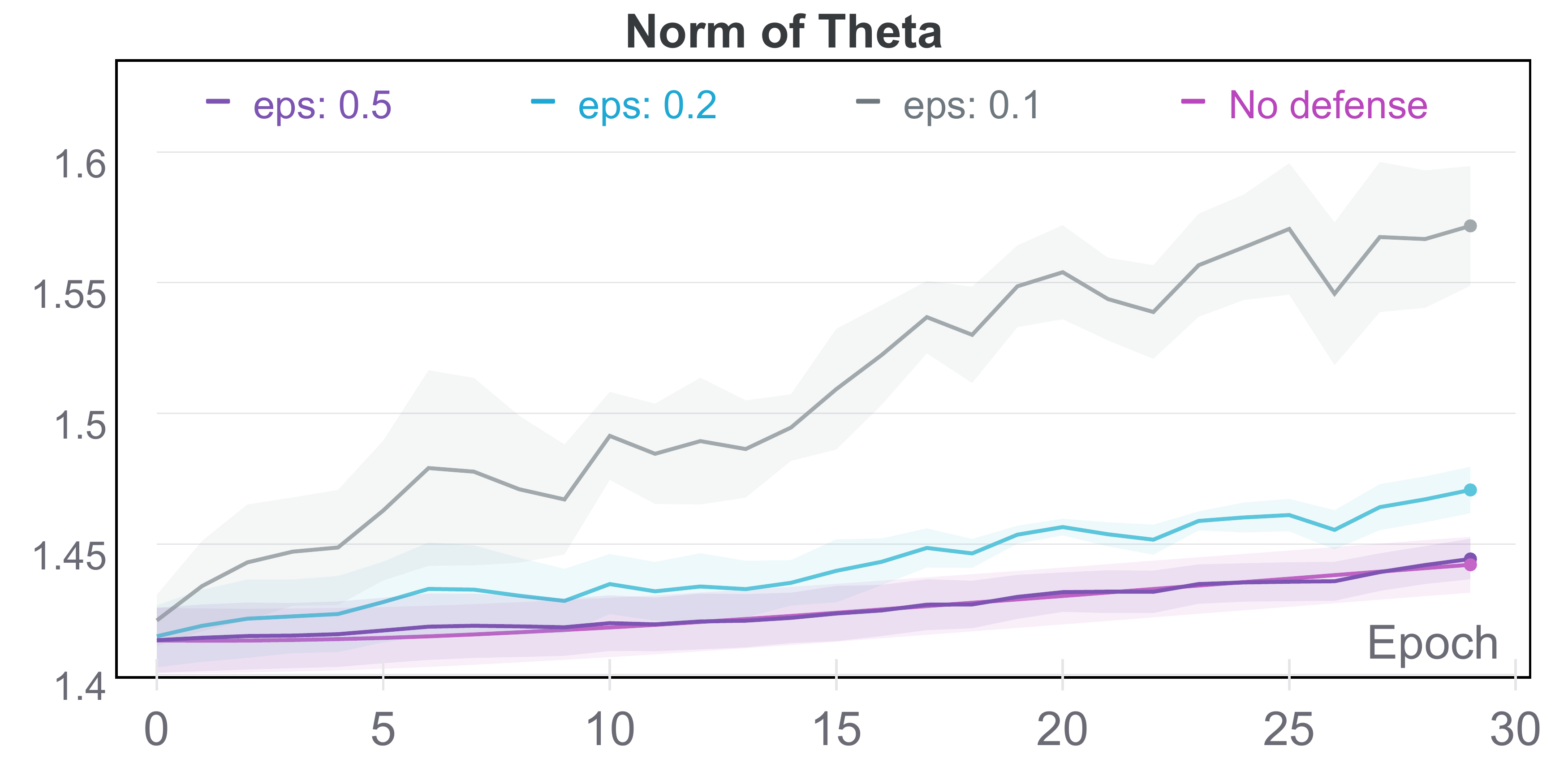}
\caption{Other results on MNIST. Top: For all choices of $\epsilon$, the success rate of label attack is reduced to about $50$\%, which is same to random guess. In contrast, for VFL without defense, the success rate of label attack is $100\%$ for several epochs. This validates the effectiveness of our attack and defense. Bottom: Norm of $\bm\theta$ is bounded by 2 throughout training, which verify the choice of $G=1$.}
\label{fig:exp_stats}
\end{figure}

\paragraph{Results on credit dataset.} In Figure~\ref{fig:exp_credit} we report the experiment results on credit dataset. The results are similar to MNIST. Since the dimension of the features in credit dataset is small, i.e., 23, the scales of Xavier Initialization and Kaiming Initialization are large. We thus use this dataset to provide a comparison between different initialization schemes. As shown in Figure~\ref{fig:exp_credit}, Xavier Initialization has slower convergence speed while the success rate of attack is lower than 1. This phenomenon partially explains why most of existing frameworks of VLR prefer to use zero initialization and verifies our analysis in Section~\ref{sec_security_A}.
\begin{figure}[h]
\includegraphics[width=0.5\textwidth]{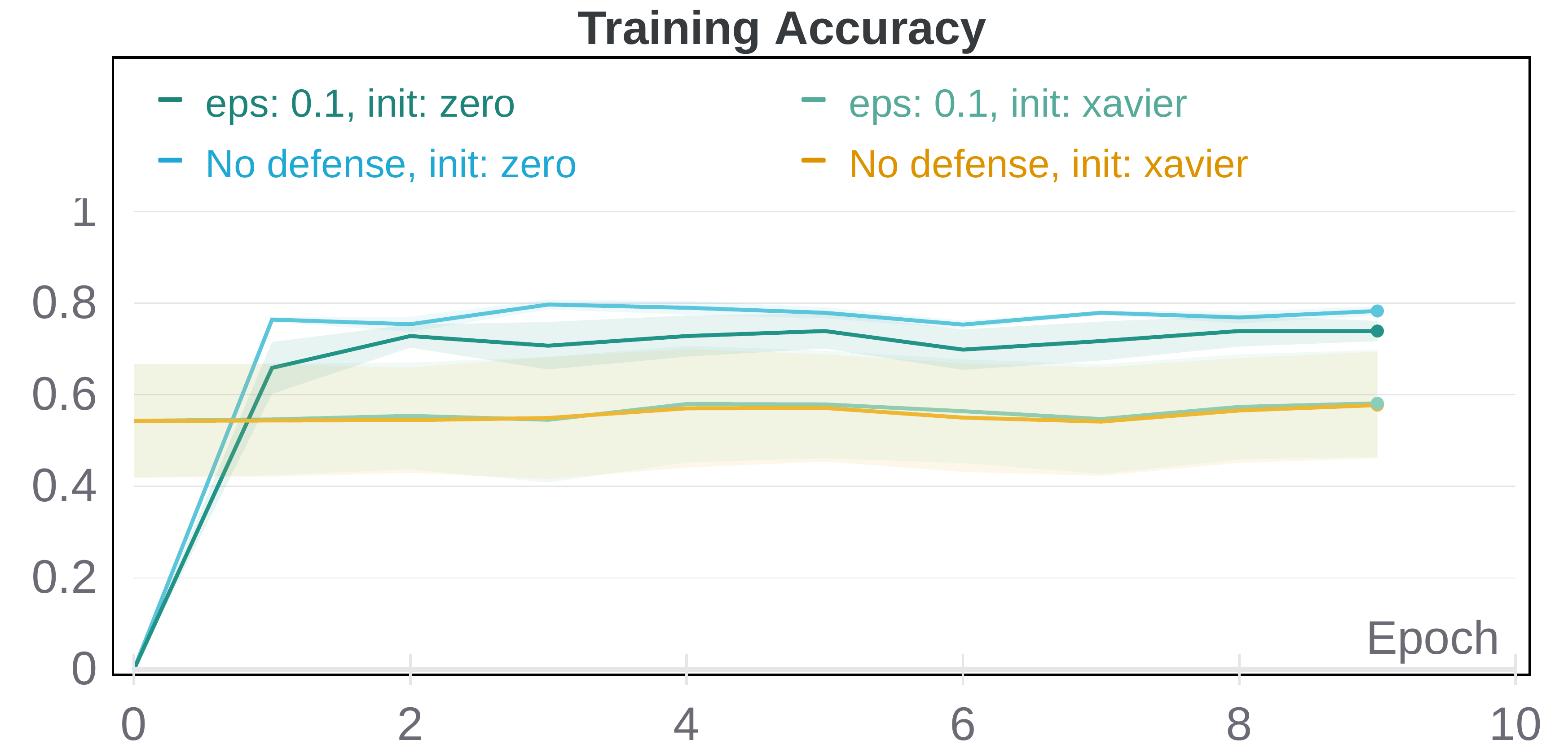}
\includegraphics[width=0.5\textwidth]{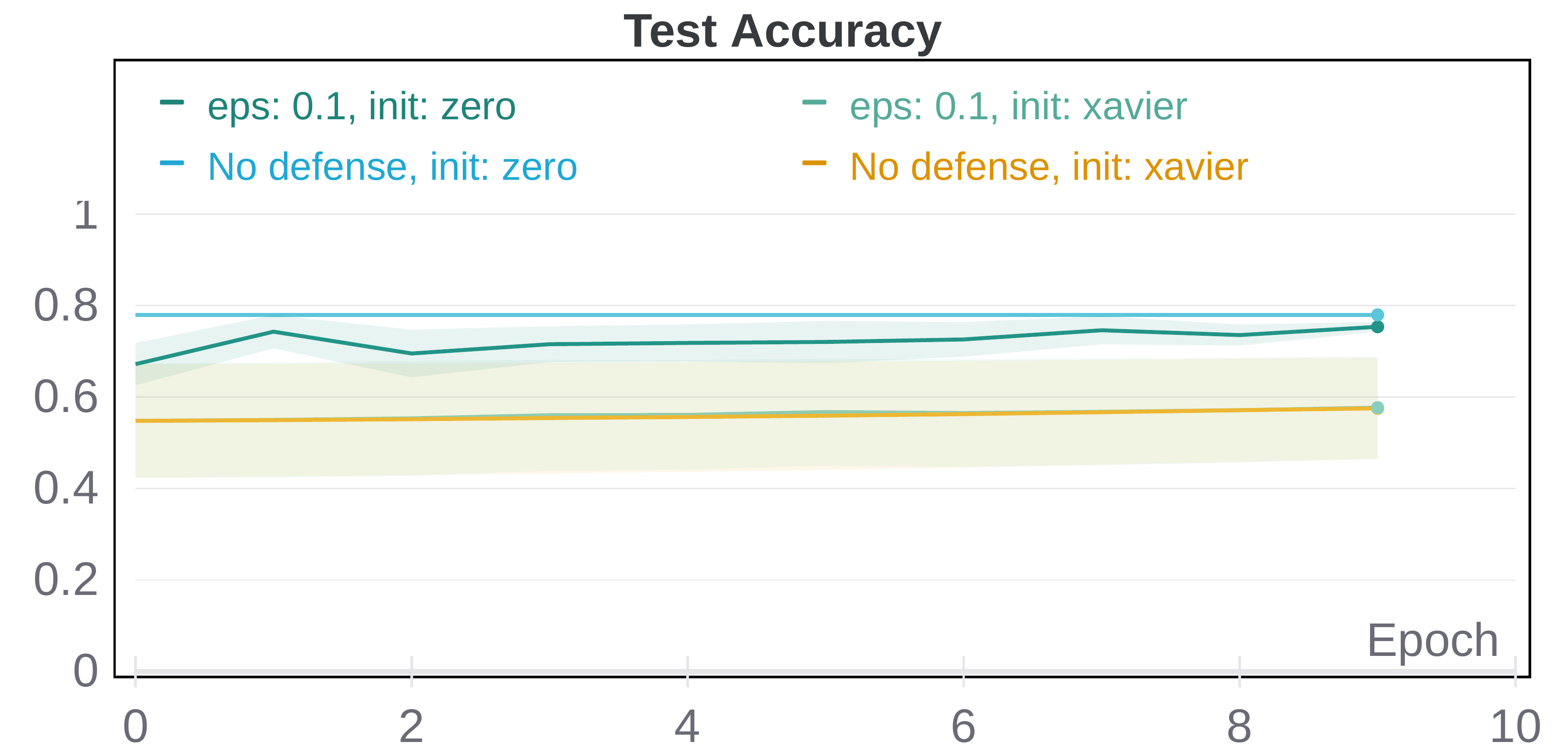}
\includegraphics[width=0.5\textwidth]{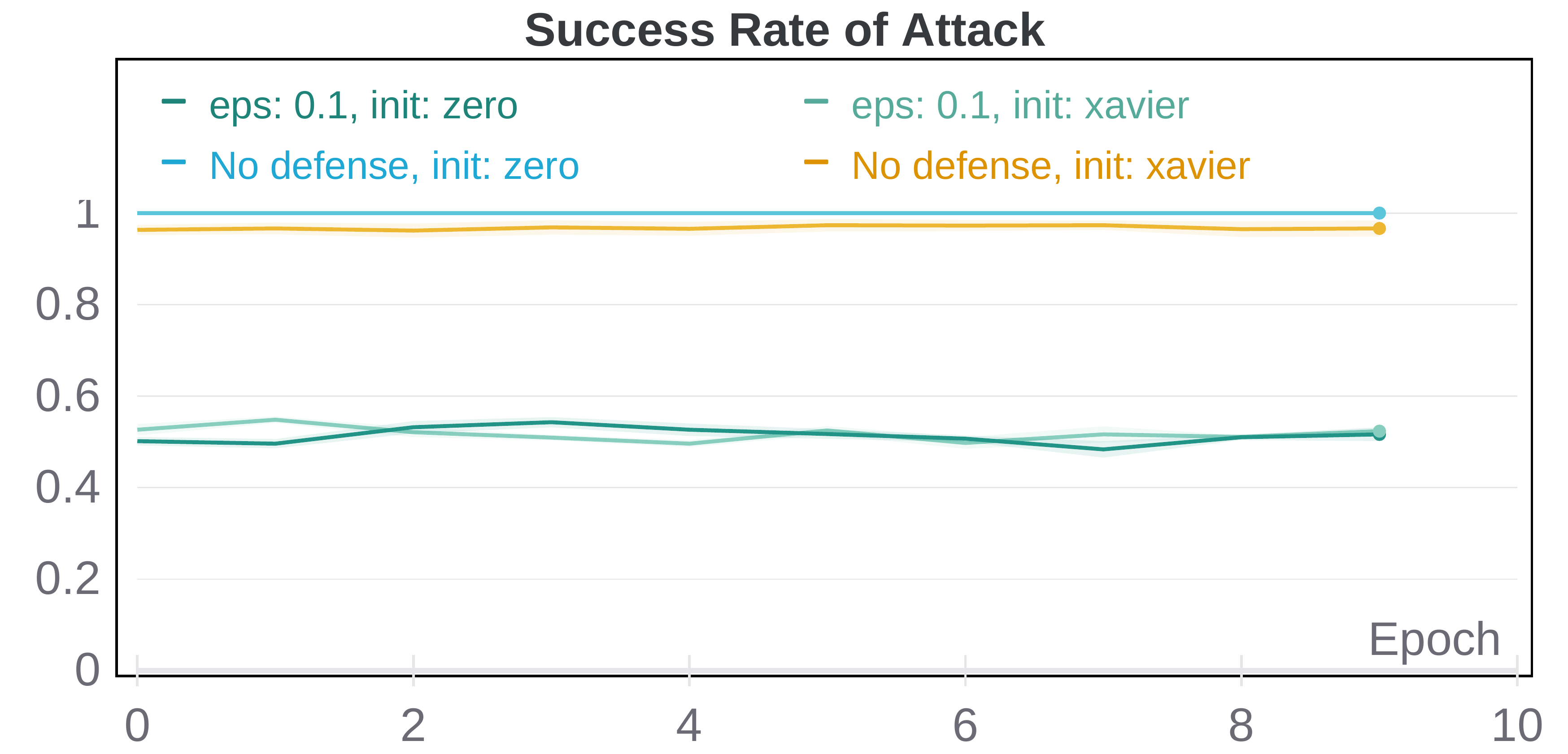}
\caption{Experiment results on credit dataset with a highlight on the comparison between different initialization schemes. Models with Xavier Initialization converge slower than those with zero initialization as shown in Top and Middle figures. The success rate of label attack is less than 1 for Xavier Initialization on credit dataset because the feature dimension is small (23 dimensions) so that the scale of Xavier Initialization is large.}
\label{fig:exp_credit}
\end{figure}



\section{Conclusion}  \label{sec_conclusion}
In this paper, we provide a formal privacy analysis for VLR, under an oracle of obtaining local gradients at each iteration within the honest-but-curious threat model. When the linear system is uniquely solvable, i.e. $s \le d_A$ or $s\le d_B$, we 1) identify feature attack for Party $A$ when $d_B = 1$ and provide hardness result when $d_B \ge 2$; 2) construct label recovery attack for Party $B$ when the initialization is small. By taking into account the protocol of obtaining gradient, which is mainly based on Homomorphic Encryption, we propose an active attack via generating and compressing auxiliary ciphertext and thereby relax the constraints of batch size. Since increasing the batch size might not be adequate to addressing the privacy leakage, we adopt a simple-yet-effective countermeasure based on DP, which injects noise to the sensitive information before communication. We provide both utility and privacy guarantees for the updated algorithm. Experiment results on benchmark datasets demonstrate the effectiveness of our attack and defense, indicating the vulnerability of existing frameworks  as well as the power of coupling DP with HE techniques. 

{
To be fair, other forms of attacks certainly exist (either simpler or being much more sophisticated) beyond those proposed in this paper. For example, it is possible that when the feature dimension is small, inference attacks using model's predictions could reveal partial information about the data. It is also possible that some advanced techniques (such as gradient inversion \cite{yin2021see}) could allow the adversary to recover much more information about the data. However, our goal in this paper is to analyze the existence of privacy leakage in a \textit{provable, mathematically-rigorous} manner. A desirable answer should either be `Yes’, as suggested in Theorem \ref{theorem precise}, \ref{theorem extreme} and \ref{theorem small}, or `No’, as pointed out by Theorem \ref{theorem hardness}. Saying in another way, we are trying to uncover the most intrinsic and fundamental vulnerability of the VLR systems that resort to the gradient oracle. We believe this is exactly what separates our paper from previous works.}

In all, our work suggests that all VFL frameworks that solely rely on HE might contain severe privacy risks, and DP, being an important building block of horizontal federated learning, can also exert its power in the vertical setting. We hypothesize that a \textit{hybrid} protocol that makes use of both HE (or MPC) and DP would lead to the best practical performance.
From a broader perspective, we hope this work would inspire more works to focus on the (rigorous) privacy analysis of general VFL beyond logistic regression, which could lead to substantial improvement over existing algorithms and protocols.


\bibliographystyle{ACM-Reference-Format}
\bibliography{ccs-sample}

\appendix

\section{Omitted Proofs from Section 3} \label{app_sub_A}

\begin{proof}[Proof of Theorem \ref{Theorem Linear}]
We use $m(E)$ to denote the Lebesgue measure of set $E \in \mathbb{R}^{d' \times n}$. Since the data are randomly sampled from a continuous probability distribution, it suffices to show the following set
\begin{align}
    L:= \bigg\{\exists \mathcal{B} \subseteq D \ \  s.t. \ \ |\mathcal{B}|= s \wedge \{\bm z_j\}_{j\in \mathcal{B}} \ \ \text{are linear dependent} \bigg\}
\end{align}
is of Lebesgue measure zero. Note 
\begin{align}
    C := \bigg\{\mathcal{B}\subseteq D: |\mathcal{B}| = s\bigg\}
\end{align}
is finite, so by taking a union bound, it suffices to show: for a \textit{fixed} set $\mathcal{B}$ satisfying $|\mathcal{B}|=s$,
\begin{align}
    E_\mathcal{B} = \bigg\{\{\bm z_j\}_{j\in \mathcal{B}} \ \ \text{are linear dependent}\bigg\}
\end{align}
is of Lebesgue measure zero. Denote the matrix concatenated by $\{\bm z_j\}_{j\in \mathcal{B}}$ in column as $\bm M \in \mathbb{R}^{d' \times s}$ and the $i$-th row of $\bm M$ to be $\bm m_{i}^{\top}$. Note that $\{\bm z_j\}_{j\in \mathcal{B}}$ are linear independent is equivalent to
\begin{align*}
     \forall 1\le i_1< \cdots < i_s \le d', \ \ s.t. \ \ \det\left(\bigg\{\bm m_{i_1}, \cdots, \bm m_{i_s}\bigg\}\right) = 0.
\end{align*}
Further, singular matrix has zero Lebesgue measure, so
\begin{align}
    m(E_\mathcal{B}) \le  m\bigg(\det\left(\bigg\{\bm m_{i_1}, \cdots, \bm m_{i_s}\bigg\}\right) = 0\bigg) = 0,
\end{align}
which gives the result as desired.
\end{proof}

\begin{proof}[Proof of Lemma \ref{lem inter}]
We say a set $G \subset \{1, 2, \cdots, n\}$ is sortable if for $i, j \in G$, we hold the value of $x_i^B/x_j^B$. The idea is to count the number of sortable set at the end of each epoch $t$, which we denote as $n_t$. Note for two sets $G_1$ and $G_2$, if $p \in G_1, q\in G_2$, and both $p,q$ appears in the same batch, then we can merge $G_1$ and $G_2$ into a larger sortable set. Eventually our goal is to obtain a single set $\{1, 2,\cdots, n\}$ via merging.

After the first epoch, there are $m$ different sets $\mathcal{B}_1, \cdots, \mathcal{B}_m$, so $n_1 = m$. Now, suppose at the end of epoch $t$, there is a set $G$ constituted of $k$ different $\mathcal{B}_i$ and $1\le k \le m-1$, then the probability of the following event
\begin{align}
    E_G = \{\forall j \in [tm+1, tm+m], \ S_j \subset G \ \vee \ S_j \cap G = \emptyset \}
\end{align}
(i.e. $G$ is regrouped into $k$ new batches in $\{\mathcal{B}_j\}_{j \in [tm+1, tm+m]}$) can be bounded as
\begin{align}
    \frac{\prod_{i=1}^k C_{is}^s \prod_{j=1}^{m-k} C_{js}^s}{\prod_{l=1}^m C_{ls}^s} \le \frac{1}{C_{ms}^s} \le \frac{1}{m^s}.
\end{align}
Here, the randomness comes from the random draw of the mini-batches. Also, when $E_G$ does not happen, by definition the set $G$ will intertwine with an element $q \notin G$. Suppose $q \in \mathcal{B}_l$ for some $l \in [m]$, then $G$ will merge with a larger group that contains $\mathcal{B}_l$ as its subset. Therefore, with failure probability at most $\frac{n_t}{m^s}$, each of the set will merge with another in the following epoch, leading to $n_{t+1} \le n_t /2$. Hence, we can obtain a single set $\{1, 2,\cdots, n\}$ in at most $T = \lfloor\log_2(m)\rfloor + 1$ epochs, and by union bound the failure probability is at most
\begin{align}
    \sum_{t=1}^T \frac{n_t}{m^s} \le \frac{2n_1}{m^s} = \frac{2}{m^{s-1}}.
\end{align}
\end{proof}

\section{Omitted proofs from Section 4} \label{app_sub_B}
\begin{proof}[Proof of Theorem \ref{theorem extreme}]
Since $d_A = 0$, we have
\begin{align}
    \bm\theta_{t-1}^{\top}\bm x_i = \left(\bm\theta^B_{ t-1}\right)^{\top}\bm x^B_{i}.
\end{align}
Therefore, by extracting the information known to himself, Party $B$ can deduce the information of
\begin{align}
    \bm I_{t} = \sum_{i \in \mathcal{B}_t} y_i \bm x^B_{i}.  
\end{align}
We will now show the \textit{constrained} linear system 
\begin{equation}
\label{eq constrain}
\left\{
\begin{aligned}
\sum_{i \in \mathcal{B}_t} z_i \bm x^B_{i} &= \bm I_{t} \\
s.t. \ \ \  z_i^2 &=1, \ \ \ \forall i\in \mathcal{B}_t
\end{aligned}
\right.
\end{equation}
with respect to $\{z_i\}_{i \in \mathcal{B}_t}$ has unique solution with probability one. The existence is trivial since $z_i = y_i, \ \forall i \in \mathcal{B}_t$ is a solution. To see the uniqueness, suppose there are two different solutions $\{z_i\}_{i \in \mathcal{B}_t}$ and $\{\tilde{z}_i\}_{i \in \mathcal{B}_t}$ satisfying Eq. (\ref{eq constrain}). Taking the difference, we have
\begin{align}
    \sum_{i \in \mathcal{B}_t}w_i\bm x^B_{i} = 0,   \label{eq homo}
\end{align}
where 
\begin{align}
   \bm w \neq\bm  0 \ \wedge \ w_i \in \{-1, 0, 1\}, \ \ \ \ \forall i \in \mathcal{B}_t. \label{viable}
\end{align}
Denote the set that contains all possible $\bm w$ in (\ref{viable}) as $T^s$. Similar to the proof in Theorem \ref{Theorem Linear}, we only need to show for a fixed set $\mathcal{B} \in D$ with $|\mathcal{B}| = s$,
\begin{align}
    E_\mathcal{B} = \bigg\{\exists\bm w \in T^s \ \  s.t. \  \ \sum_{i \in \mathcal{B}}w_i\bm x^B_{i} = 0 \bigg\}   \label{eq set}
\end{align}
is of Lebesgue measure zero. Since $T^s$ is a finite set with cardinality $3^s-1$, it suffices to show for a fixed $\bm w \in T^s$ that
\begin{align}
    L_\mathcal{B} = \bigg\{\sum_{i \in \mathcal{B}}w_i\bm x^B_{i} = 0\bigg\}
\end{align}
satisfies $m(L_\mathcal{B}) =0$. This is obviously true since the data $\{\bm x^B_{i}\}_{i\in [n]}$ is drawn from a continuous probability distribution, and that $L_\mathcal{B}$ is a proper subspace of $\mathbb{R}^{d_B \times s}$. 

\end{proof}

\section{Omitted proof from Section 6} \label{app_sub_C}

\begin{proof}[Proof of Lemma \ref{lemma gradient error}]
By considering a union bound, it suffices to show
\begin{align}
    \|\bm e_t\|_2 \le \sqrt{\frac{k}{s}}\max\{\sigma_A, \sigma_B\}
\end{align}
holds with probability at least $1 - 2\exp(-ck)$ for a fixed $t$. Note
\begin{align}
    \|\bm e_t\|_2 \le \frac{1}{4s}\left(\|\bm X_{\mathcal{B}_t}^A\bm Z^B\|_2 + \|\bm X_{\mathcal{B}_t}^B\bm Z^A\|_2\right),
\end{align}
so by symmetry it suffices to show
\begin{align}
    \|\bm X_{\mathcal{B}_t}^A\bm Z^B\|_2 \le 2\sqrt{ks\sigma_B^2}
\end{align}
with probability at least $1-\exp(-ck)$. Define $\bm P_A := \left(\bm X_{\mathcal{B}_t}^A\right)^{\top}\bm X_{\mathcal{B}_t}^A$, then we have $\|\bm P_A\|_F^2 \le s^2$ and $\|\bm P_A\|_2 \le \|\bm P_A\|_F \le s$. Applying Hanson-Wright inequality (Theorem 6.2.1 in \cite{vershynin2018high}), we have for every $t >0$,
\begin{align}
    & \ \ \ \ \mathbb{P}\bigg\{\left(\bm Z^B\right)^{\top}\bm P_A \bm Z^B - \mathbb{E}\left[\left(\bm Z^B\right)^{\top}\bm P_A \bm Z^B\right] \ge t \bigg\}  \nonumber \\
    &\le \exp\left[-c\min\left(\frac{t^2}{9s^2\sigma_B^4},\frac{t}{3s\sigma_B^2}\right)\right].
\end{align}
Pick $t = 3ks\sigma_B$, then with probability at least $1-\exp(-ck)$, we have
\begin{align}
    \left(\bm Z^B\right)^{\top}\bm P_A \bm Z^B &\le 3ks\sigma_B^2 + \mathbb{E}\left[\left(\bm Z^B\right)^{\top}\bm P_A \bm Z^B\right] \\
    &= 3ks\sigma_B^2 + \sigma_B^2\Tr(\bm P_A) \\
    &\le 4ks\sigma_B^2,
\end{align}
This finishes the proof as desired.
\end{proof}

\section{Heuristic $(\epsilon, \delta)$-DP Guarantee of Algorithm 3}  \label{app.exp}
We will show that Algorithm \ref{alg3} satisfies $(\epsilon, \delta)$-DP (though we do not see this as a formal privacy guarantee) when the variance of the Gaussian random variables are chosen according to Eq. (\ref{DP_A}) and (\ref{DP_B}).
We begin with a useful definition.
\begin{definition}[$\ell_2$-sensitivity]
The $\ell_2$-sensitivity of a function $f: \mathcal{D} \to \mathbb{R}^d$ is defined as:
\begin{align}
    \Delta_2(f) := \max\limits_{D_1, D_2} \|f(D_1) - f(D_2)\|_2
\end{align}
for all $D_1, D_2 \in \mathcal{D}$ that differ by at most one instance.
\end{definition}
Denote $f$ as the original algorithm given by Eq. (\ref{SGD}). In the following analysis, we will use $\Delta(\bm v)$ to represent $\Delta_2(f_{\bm v})$, where $f_{\bm v}$ repeats the procedure of $f$ until it outputs the (intermediate) result $\bm v$. 

For arbitrary dataset $D$ and $D'$ which differ by at most one instance, consider the corresponding sequences $\{\bm \theta_{t-1}\}_{t\in [T]}$ and $\{\bm \theta_{t-1}'\}_{t\in [T]}$ generated via feeding $D$ and $D'$ to $f$, and the mini-batch at step $t$ are $\mathcal{B}_t$ and $\mathcal{B}_t'$ respectively. Our first step is to give a recursive relation on $\Delta(\bm \theta_{t-1})$.

\begin{lemma} \label{sensitive recur}
Let $\eta \le 1$. Then for $t \in [T]$, we have
\begin{align}
    \Delta(\bm \theta_{t}) \le \begin{cases}
      \Delta(\bm \theta_{t-1}) &  \emph{if} \ \ \mathcal{B}_t = \mathcal{B}_t'   \\
      \Delta(\bm \theta_{t-1}) + \frac{2\sqrt{2}\eta G}{s} &  \emph{o.w.}
    \end{cases}
\end{align}
\end{lemma}

\begin{proof}[Proof of Lemma \ref{sensitive recur}]
When $\mathcal{B}_t = \mathcal{B}_t'$, we have
\begin{align}
    & \ \ \ \ \|\bm \theta_{t} - \bm \theta_{t}'\|_2^2 - \|\bm \theta_{t-1} - \bm \theta_{t-1}'\|_2^2 \\
    &= \eta^2 \|\nabla\mathcal{L}_{\mathcal{B}_t}(\bm\theta_{t-1}) - \nabla\mathcal{L}_{\mathcal{B}_t}(\bm\theta_{t-1}')\|_2^2  \nonumber \\ & \ \ \ \  - 2\eta \langle \bm \theta_{t-1} - \bm \theta_{t-1}', \nabla\mathcal{L}_{\mathcal{B}_t}(\bm\theta_{t-1}) - \nabla\mathcal{L}_{\mathcal{B}_t}(\bm\theta_{t-1}') \rangle \\
    &\le - \left(\frac{2\eta}{\beta} -\eta^2\right) \|\nabla\mathcal{L}_{\mathcal{B}_t}(\bm\theta_{t-1}) - \nabla\mathcal{L}_{\mathcal{B}_t}(\bm\theta_{t-1}')\|_2^2 \le 0, 
\end{align}
where we apply the co-coercivity of gradient in the inequality, since $\mathcal{L}_{\mathcal{B}_t}$ is convex and $\frac{1}{4}$-smooth.

When $\mathcal{B}_t \neq \mathcal{B}_t'$, they differ by at most one instance. Denote their intersection as $\mathcal{I}_t$ and $p = \mathcal{B}_t \backslash \mathcal{I}_t, \ q = \mathcal{B}_t' \backslash \mathcal{I}_t$. Now denote
\begin{align}
    \bar{\mathcal{L}}_t(\bm \theta) = \frac{s-1}{s}\mathcal{L}_{\mathcal{I}_t}(\bm \theta),
\end{align}
then
\begin{align}
    & \ \ \ \ \|\bm \theta_{t} - \bm \theta_{t}'\|_2^2 \\
    &= \|\bm \theta_{t-1} - \bm \theta_{t-1}'\|_2^2 - 2\eta\bigg< \bm \theta_{t-1} - \bm \theta_{t-1}', \nabla \bar{\mathcal{L}}_t(\bm \theta_{t-1}) - \nonumber \\
    & \ \ \ \   \nabla \bar{\mathcal{L}}_t(\bm \theta_{t-1}') +  \frac{1}{s}\left(\nabla\mathcal{L}_p(\bm \theta_{t-1}) -\nabla\mathcal{L}_q(\bm \theta_{t-1}') \right)    \bigg> + \nonumber \\
    & \ \ \ \ \eta^2 \bigg\| \nabla \bar{\mathcal{L}}_t(\bm \theta_{t-1}') - \nabla \bar{\mathcal{L}}_t(\bm \theta_{t-1}) + \nonumber \\  &  \ \ \ \ \ \  \ \ \frac{1}{s}\left(\nabla\mathcal{L}_p(\bm \theta_{t-1}) -\nabla\mathcal{L}_q(\bm \theta_{t-1}') \right) \bigg\|_2^2 \\
    &\le \|\bm \theta_{t-1} - \bm \theta_{t-1}'\|_2^2  + \frac{4\eta G}{s}\|\bm \theta_{t-1} - \bm \theta_{t-1}'\|_2 + \frac{8\eta^2 G^2}{s^2} \nonumber \\
    &-\left(\frac{2\eta}{\beta}-2\eta^2\right)\|\nabla \bar{\mathcal{L}}_t(\bm \theta_{t-1}) - \nabla \bar{\mathcal{L}}_t(\bm \theta_{t-1}')\|_2^2 \\
    &\le \left(\|\bm \theta_{t-1} - \bm \theta_{t-1}'\|_2 + \frac{2\sqrt{2}\eta G}{s} \right)^2.
\end{align}
Since the choice on $D$ and $D'$ is arbitrary, we obtain the result as desired.
\end{proof}

We now proceed to control $\Delta(\bm{SV}_t^A)$ and $\Delta(\bm{SV}_t^B)$.
\begin{lemma}  \label{lemma sv}
Let $\eta \le 1$. Then for $t \in [T]$, we have
\begin{align}
    \Delta(\bm{SV}_{t}^A) &\le \begin{cases}
      \sqrt{s}\Delta(\bm \theta_{t-1}) &  \emph{if} \ \ \mathcal{B}_t = \mathcal{B}_t'   \\
       \sqrt{s\Delta^2(\bm \theta_{t-1}) + (8G)^2} &  \emph{o.w.}
    \end{cases}, \\  
    \Delta(\bm{SV}_{t}^B) &\le \begin{cases}
      \sqrt{s}\Delta(\bm \theta_{t-1}) &  \emph{if} \ \ \mathcal{B}_t = \mathcal{B}_t'   \\
       \sqrt{s\Delta^2(\bm \theta_{t-1}) + (8G-4)^2} &  \emph{o.w.}
    \end{cases}.
\end{align}

\end{lemma}

\begin{proof}[Proof of Lemma \ref{lemma sv}]
When $\mathcal{B}_t = \mathcal{B}_t'$, we have
\begin{align}
     \|\bm{SV}_{t}^A - (\bm{SV}_{t}^A)'\|_2 
    &= \sqrt{\sum_{i \in \mathcal{B}_t} \big< \bm\theta_{t-1}^A-(\bm\theta_{t-1}^A)', \bm x_i^A \big>^2  } \\
    &\le \sqrt{\sum_{i \in \mathcal{B}_t} \| \bm\theta_{t-1}-\bm\theta_{t-1}'\|_2^2 \|\bm x_i\|_2^2 } \\
    &\le \sqrt{s}\| \bm\theta_{t-1}-\bm\theta_{t-1}'\|_2,
\end{align}
so $\Delta(\bm{SV}_{t}^A) \le \sqrt{s}\Delta(\bm \theta_{t-1})$. Similarly, $\Delta(\bm{SV}_{t}^B) \le \sqrt{s}\Delta(\bm \theta_{t-1})$. On the other hand, when $\mathcal{B}_t \neq \mathcal{B}_t'$ we have
\begin{align}
    & \ \ \ \|\bm{SV}_{t}^A - (\bm{SV}_{t}^A)'\|_2 \\
    &= \sqrt{
    \splitfrac{
\sum_{i \in \mathcal{I}_t} \big< \bm\theta_{t-1}^A-(\bm\theta_{t-1}^A)', \bm x_i^A \big>^2}{
+ \left(\langle \theta_{t-1}^A, \bm x_p^A  \rangle -2y_p - \langle \theta_{t-1}^A, \bm x_q^A  \rangle + 2y_q \right)^2}\,} \\
    &\le \sqrt{\sum_{i \in \mathcal{I}_t} \| \bm\theta_{t-1}-\bm\theta_{t-1}'\|_2^2 \|\bm x_i\|_2^2 + (8G)^2 } \\ 
    &\le \sqrt{s\| \bm\theta_{t-1}-\bm\theta_{t-1}'\|_2^2+ (8G)^2}.
\end{align}
Therefore, $\Delta(\bm{SV}_{t}^A) \le \sqrt{s\Delta^2(\bm \theta_{t-1}) + (8G)^2}$. Similarly, $\Delta(\bm{SV}_{t}^B) \le \sqrt{s\Delta^2(\bm \theta_{t-1}) + (8G-4)^2}$.
\end{proof}

Finally, we will resort to the following lemma, which relates $\ell_2$-sensitivity to a DP guarantee and is standard in literature.

\begin{lemma}[Theorem A.1 of \cite{dwork2014algorithmic}] \label{lemma GM}
Let $f$ be an arbitrary function generating $d$-dimensional
outputs and $\epsilon \in (0, 1)$. Then for $c^2 > 2\ln(1.25/\delta)$, the Gaussian Mechanism with parameter $\sigma \ge  c\Delta_2 f/\epsilon$ is $(\epsilon,\delta )$-differentially private.
\end{lemma}

The final result is stated in the following theorem.
\begin{theorem}   \label{theorem DP}
Let
\begin{align}
    \sigma_A = \sqrt{2\ln\left(\frac{5}{4\delta}\right)}\frac{\sqrt{\frac{8G^2e^2T\eta^2}{s} + 64G^2e}}{\epsilon},
\end{align}
\begin{align}
    \sigma_B = \sqrt{2\ln\left(\frac{5}{4\delta}\right)}\frac{\sqrt{\frac{8G^2e^2T\eta^2}{s} + (8G-4)^2e}}{\epsilon},
\end{align}
then Algorithm \ref{alg3} with step size $\eta \le 1$ is $(\epsilon, \delta)$-differentially private w.r.t. $\{Sec(\bm{SV}_t^A)\}_{t \in [T]}$ and $\{Sec(\bm{SV}_t^B)\}_{t \in [T]}$.
\end{theorem}

\begin{proof}[Proof of Theorem \ref{theorem DP}]
Since there exists at most one $t$ in an epoch such that $\mathcal{B}_t \neq \mathcal{B}_t'$, we have
\begin{align}
    \Delta(\{\bm{SV}_t^A\}_{t\in [T]}) &\le \sqrt{\sum_{t\in [T]} \Delta^2(\bm{SV}_t^A)} \\
    &\le \sqrt{\sum_{t \in [T]} s\Delta^2(\bm \theta_{t-1}) + (8G)^2e} \\
    &\le \sqrt{sm\sum_{i \in [e]}\left(\frac{2\sqrt{2}i\eta G}{s}\right)^2 + (8G)^2e} \\
    &\le \sqrt{\frac{8G^2e^2T\eta^2}{s} + 64G^2e},
\end{align}
where we use Lemma \ref{lemma sv} in the second inequality, Lemma \ref{sensitive recur} and the fact $\Delta(\bm \theta_0) = 0$ in the third inequality. 
Similarly, we have
\begin{align}
    \Delta(\{\bm{SV}_t^B\}_{t\in [T]}) \le\sqrt{\frac{8G^2e^2T\eta^2}{s} + (8G-4)^2e}.
\end{align}
Applying Lemma \ref{lemma GM}, we can conclude that Algorithm \ref{alg3} is $(\epsilon, \delta)$-differentially private w.r.t. $\{Sec(\bm{SV}_t^A)\}_{t \in [T]}$ and $\{Sec(\bm{SV}_t^B)\}_{t \in [T]}$. 
when $\Sigma_A$ and $\Sigma_B$ are chosen according to Eq. (\ref{DP_A}) and (\ref{DP_B}).
\end{proof}

\section{Discussion of Minimax Approximation} \label{app_minimax}
Let $P_d$ denote the set of polynomials of degree at most $d$, and for a continuous
function $f \in C[a, b]$ denote 
\begin{align}
    \|f\|_{\infty} = \max\big\{|f(x)|: x \in [a, b]\big\},
\end{align}
the minimax approximation is then defined as follow.
\begin{definition}
$p \in P_d$ is a $d$-th minimax approximation of $f \in C[a, b]$ if
\begin{align}
    \|p - f\|_{\infty} = \inf\limits_{q \in P_d}\big\{\|q-f\|_{\infty}\big\}.
\end{align}
\end{definition}
Consider the $3$-th minimax approximate of the sigmoid function of $[-5,5]$, which is shown to be 
\begin{align}
    \sigma_3(z) = -0.004z^3 + 0.197 z + 0.5
\end{align}
in \cite{chen2018logistic}. The mini-batch gradient can then be approximated as
\begin{align}
 \nabla\ell_{\mathcal{B}_t}(\bm\theta_{t-1}) \approx \frac{1}{s}\sum_{i\in \mathcal{B}_t} f_{i,t}\bm x_i   
\end{align}
at $t$-th iteration, where the coefficient is defined through
\begin{align}
    f_{i,t} := -0.004(\bm\theta_{t-1}^{\top}\bm x_i)^3 + 0.197\bm\theta_{t-1}^{\top}\bm x_i -0.5y_i.
\end{align}
We assume $s \le \min\{d_A, d_B\}$ in the following analyses, so that both parties can learn the precise value of $f_{i,t}$ with probability one by Theorem \ref{Theorem Linear}.

\subsection{Privacy Analysis for Party $B$}
Decompose
    \begin{align}
    \bm\theta_{t-1}^{\top}\bm x_i = \underbrace{\left(\bm\theta^A_{t-1}\right)^{\top}\bm x^A_{i}}_{g^A_{i,t}} + \underbrace{\left(\bm\theta^B_{t-1}\right)^{\top}\bm x^B_{i}}_{g^B_{i,t}}.
\end{align}
So by extracting the information known to himself (the part involving $g_{i,t}^A$ and $y_i$) from $f_{i,t}$, Party $A$ can obtain the information of
\begin{align}
    r_{i,t} &= -0.004\left(g^B_{i,t}\right)^3 - 0.012g^A_{i,t}\left(g^B_{i,t}\right)^2  \nonumber \\
    & \ \ \ - \left(0.012\left(g^A_{i,t}\right)^2-0.197\right)g^B_{i,t},
\end{align} 
which is a cubic polynomial with respect to the variable $g_{i,t}^B$. Consider the corresponding function
\begin{align}
    m_{i, t}(z) = 0.004z^3 + 0.012g^A_{i,t}z^2 + \left(0.012\left(g^A_{i,t}\right)^2-0.197\right)z,
\end{align}
it is straightforward to see that $-m_{i,t}(z)$ has a local maximum at $z_{0} =\sqrt{\frac{197}{12}}-g^A_{i,t}$ and a local minimum at $z_{1}=-\sqrt{\frac{197}{12}}-g^A_{i,t}$. Therefore,  $A$ can uniquely determine $g^B_{i,t}$ from $r_{i,t}$ iff
\begin{align*}
    r_{i,t} \notin [-m_{i,t}(z_{1}),-m_{i,t}(z_{0})].
\end{align*}
Denote
\begin{align}
    \mathcal{C}_t = \{i\in \mathcal{B}_t: r_{i,t} \notin [-m_{i,t}(z_{1}),-m_{i,t}(z_{0})]\},
\end{align}
then under $d_B = 1$, Party $A$ can sort $\mathcal{C}_t$ according to the $x_B$ feature. If we assume the full dataset $D$ becomes sortable in $Q$ epochs, then using a similar approach as in Theorem \ref{theorem precise}, we obtain the following theorem. 

\begin{theorem}
In the minimax approximation scheme, suppose $s \le d_A$ and $d_B= 1$, and $D$ becomes sortable in $Q$ epochs. Then Party $A$ can learn the precise value of the data in terms of feature $x_B$(up to a difference in sign) in $\max\{Q,2\}$ epochs.
\end{theorem} 

\subsection{Privacy Analysis for Party $A$}
In analogy to Theorem \ref{theorem extreme}, Party $B$ can determine the label within every batch with probability one if $d_A = 0$ (note this does not require $s \le d_B$). Now suppose $d_A \ge 1$ and $s \le d_B$. Note
\begin{align}
    f_{i,t}\mid_{\bm\theta_{t-1}^{\top}\bm x_i=0} = -0.5y_i,
\end{align}
so the intuition in Section \ref{sec_security_B} holds and the analyses naturally breaks into two parts. It is straightforward to see that the first part only involves the initialization of $\bm \theta_0$ and is irrelevant to the approximation method, so we will mainly focus on the second part. First, consider 
\begin{align}
    n(z) = -0.004z^3 + 0.197z,
\end{align}
then it is straightforward to verify that 
\begin{align*}
    |n(z)| < 0.5, \ \ \forall |z| < 3.2.
\end{align*}
This implies that Criterion \ref{criterion} is valid at iteration $t$ if 
\begin{align}
    \max\limits_i \{|\bm \theta_{t-1}^{\top}\bm x_i|\} \le 3.2.
\end{align}
Now, assuming a constant learning rate $\eta_t = \eta$,  we have 
\begin{align}
    \bm\theta_t^{\top}\bm x_i = \bm\theta_{t-1}^{\top}\bm x_i - \frac{\eta}{s} \sum_{j \in \mathcal{B}_t} [&-0.004(\bm\theta_{t-1}^{\top}\bm x_j)^3 + 0.197\bm\theta_{t-1}^{\top}\bm x_j \nonumber \\
    &-0.5y_j]\bm x_j^{\top}\bm x_i.  
\end{align}
Applying Triangle Inequality and Cauchy-Schwarz Inequality, 
\begin{align}
    |\bm\theta_{t}^{\top}\bm x_i| \le |\bm\theta_{t-1}^{\top}\bm x_i| +  \frac{\eta}{s}\sum_{j\in \mathcal{B}_t}[0.004|\bm\theta_{t-1}^{\top}\bm x_j|^3 
    + 0.197|\bm\theta_{t-1}^{\top}\bm x_j| + 0.5]. 
\end{align}
Suppose $\max\limits_{i} \{|\bm\theta_{t-1}^{\top}\bm x_i|\} < 3.2$ and taking maximum over both side, we have
\begin{align}
    \max_{i}\{|\bm\theta_{t}^{\top}\bm x_i|\} \le  \left(1+\frac{\eta}{4}\right)\max_{i} \{|\bm\theta_{t-1}^{\top}\bm x_i|\} + \frac{\eta}{2}.
\end{align}
Therefore, $h_t := \max\limits_{i}\{|\bm\theta_{t-1}^{\top}\bm x_i| + 2\}$ satisfies
\begin{align}
    h_{t+1} \le \left(1+\frac{\eta}{4}\right)h_{t}.
\end{align}
Note now the threshold becomes $3.2+2 = 5.2$, and we immediately have the following theorem.

\begin{theorem} \label{theorem mini_zero}
In the minimax approximation scheme, If $\bm \theta_0$ is initialized (with high probability) such that
\begin{align}
    \max_i \{|\bm \theta_0^{\top}\bm x_i|\} = \epsilon <2,
\end{align} 
then Criterion \ref{criterion} is valid for 
\begin{align}
    \overline{T_{\epsilon}}= \bigg\lceil \log_{\left(1+\frac{\eta}{4}\right)}\left(\frac{5.2}{2+\epsilon}\right) \bigg\rceil 
\end{align}
(with high probability).
\end{theorem}

\section{Discussion of Piecewise Approximation} \label{app_piecewise}
Another commonly used approximation method of the sigmoid function is the piecewise linear approximation \cite{mohassel2017secureml}:
\begin{equation}
\label{eq piecewise}
S(z) \approx 
\begin{cases}
    0,  &  z \le -0.5 \\
    z+0.5,  &  |z| < 0.5 \\
    1, &   z \ge 0.5
\end{cases}
.
\end{equation}
The mini-batch gradient can then be approximated as
\begin{align}
    \nabla\ell_{\mathcal{B}_t}(\bm\theta_{t-1})  \approx \frac{1}{s}\sum_{i\in \mathcal{B}_t}f_{i,t}\bm x_i
\end{align}
at $t$-th iteration, where the coefficient is defined through
\begin{align}
    f_{i,t} := &\vmathbb{1}[|\bm\theta_{t-1}^{\top}\bm x_i| < 0.5](\bm\theta_{t-1}^{\top}\bm x_i -0.5y_i) \nonumber \\
    +&\vmathbb{1}[y_i\bm\theta_{t-1}^{\top}\bm x_i \le -0.5] (-y_i).
\end{align}
We assume $s\le \min\{d_A,d_B\}$ in the following analyses, so that both parties can learn the precise value of $f_{i,t}$ with probability one by Theorem \ref{Theorem Linear}.

\subsection{Privacy Analysis for Party $B$}
Note that 
    \begin{equation}
        |f_{i,t}|=
        \begin{cases}
            1  & \  \text{if} \ \  y_i\bm\theta_{t-1}^{\top}\bm x_i \le -0.5 \\
    |\bm\theta_{t-1}^{\top}\bm x_i - 0.5y_i| \in (0, 1)  & \  \text{if} \ \ |\bm\theta_{t-1}^{\top}\bm x_i| < 0.5 \\
    0  & \  \text{if} \ \ y_i\bm\theta_{t-1}^{\top}\bm x_i \ge 0.5
        \end{cases}
    \end{equation}
    Decompose
    \begin{align}
    \bm\theta_{t-1}^{\top}\bm x_i = \underbrace{\left(\bm\theta^A_{t-1}\right)^{\top}\bm x^A_{i}}_{g^A_{i,t}} + \underbrace{\left(\bm\theta^B_{t-1}\right)^{\top}\bm x^B_{i}}_{g^B_{i,t}}.
\end{align}
then there are three possible cases:
    \begin{itemize}
        \item[(1)] $|f_{\bm i,t}| = 1.$ Then Party $A$ can learn the \textit{range} of $g_{i,t}^B$
        \begin{equation}
            g_{i,t}^B
            \begin{cases}
                    \le -0.5-g_{i,t}^A, & \ \ \text{if} \ \ y = 1, \\
                    \ge 0.5- g_{i,t}^A, & \ \ \text{if} \ \ y = -1.
            \end{cases}
        \end{equation}
        
        \item[(2)] $0 <|f_{i,t}| < 1.$ Then Party $A$ can learn the \textit{precise value} of $g_{i,t}^B=f_{i, t}+0.5y_i-g_{i,t}^A$.        
        \item[(3)] $f_{i,t}=0.$ Then Party $A$ can learn the \textit{range} of $g_{i, t}^B$
        \begin{equation}
            g_{i, t}^B
            \begin{cases}
                    \ge 0.5-g_{i, t}^A, & \ \ \text{if} \ \ y = 1, \\
                    \le -0.5- g_{i, t}^A, & \ \ \text{if} \ \ y = -1.
            \end{cases}
        \end{equation}
    \end{itemize}
    Therefore, if Party $A$ collect those instance which falls into the second case:
    \begin{align}
        D_t = \{i \in \mathcal{B}_t: 0<|f_{i,t}|<1\},
    \end{align}
    then under $d_B = 1$, he can sort the data in $D_t$ according to the $x_B$ feature. Also, similar to the previous method, we have the following theorem. 
    
    \begin{theorem}
In the piecewise approximation scheme, suppose $s \le d_A$ and $d_B= 1$, and $D$ becomes sortable in $R$ epochs. Then Party $A$ can learn the precise value of the data in terms of feature $x_B$ (up to a difference in sign) in $\max\{R,2\}$ epochs.
\end{theorem} 

\subsection{Privacy Analysis for Party $A$} 
In analogy to Theorem \ref{theorem extreme}, Party $B$ can determine the label within every batch with probability one if $d_A = 0$ (note this does not require $s \le d_B$). Now suppose $d_A \ge 1$ and $s \le d_B$.
    
   Similar to the previous subsection, there are three possible cases:
    \begin{enumerate}
        \item[(1)] $|f_{i,t}|=1.$ Then $f_{i,t} = -y_i$ and Party $B$ can directly obtain the label via 
        \begin{align}
            y_i = -\sign{(f_{i,t})}; \label{case 1}
        \end{align}
        \item[(2)] $0 < |f_{i,t}| < 1.$ Then $f_{i,t} = \bm\theta_{t-1}^{\top}\bm x_i - 0.5y_i$ and $|\bm\theta_{t-1}^{\top}\bm x_i| < 0.5$. Therefore, Party $B$ can directly obtain the label via
        \begin{align}
            y_i = -\sign{(f_{i,t})};  \label{case 2}
        \end{align}
        \item[(3)] $f_{i,t} = 0.$ Then Party $B$ only knows
        \begin{align}
            y_i = \sign{(\bm \theta_{t-1}^{\top}\bm x_i)} \ \ \wedge \ \ |\bm \theta_{t-1}^{\top}\bm x_i| \ge 0.5,
        \end{align}
        so he cannot say much about the label in general.  \\ 
    \end{enumerate}
        
        Therefore, denote
        \begin{align}
            \mathcal{L}_t := \{i\in \mathcal{B}_t: 0 <|f_{i,t}|\le 1\},
        \end{align}
        and we have the following theorem.
        
        \begin{theorem} \label{theorem piecewise_label}
In the piecewise approximation scheme, Party $B$ can use $y_i = -\sign(f_{i,t})$ to determine the label for $i\in \mathcal{L}_t$ and all $t$. 
\end{theorem}

\end{document}